%% file: main.tex
\documentclass{article}[11pt]
\input{header/usepackages}
\input{header/commands}

\input{header/comments}
\usepackage[a4paper, left=1in, right=1in, top=1in, bottom=1in]{geometry}

\title{Optimal Purity Estimation with Incoherent Measurements}

\author{Junseo Lee\thanks{Email: \href{mailto:junseolee@fas.harvard.edu}{junseolee@fas.harvard.edu}} \\ Harvard University \\ 
\and Chirag Wadhwa\thanks{Email: \href{mailto:chirag.wadhwa@ed.ac.uk}{chirag.wadhwa@ed.ac.uk}} \\ University of Edinburgh}
\date{}

\makeatletter
\renewcommand\paragraph{\@startsection{paragraph}{4}{\z@}%
{1.8ex \@plus1ex \@minus.2ex}%
{-1em}%
{\normalfont\normalsize\bfseries}}
\makeatother

\begin{document}

\maketitle

\disablecomments

\begin{abstract}
    In this work, we consider the fundamental task of estimating the purity of an unknown state by measuring its copies. Prior work has established tight bounds for this task in the setting of additive error estimation, but such estimates are not particularly informative when the state's purity is itself much smaller than the desired precision. Instead, we consider estimating to within $\textit{multiplicative}$ error $\varepsilon$. When one can perform general collective measurements, $\Theta\left(\frac{\sqrt{d}}{\varepsilon^2} + \frac{d}{\varepsilon}\right)$ copies are known to be necessary and sufficient for multiplicative-error purity estimation~\cite{AISW20}. However, implementing collective measurements on such a large number of copies can be experimentally demanding, and we thus aim to characterize the copy complexity of this problem with \emph{incoherent} measurements.

    In this setting, the only non-trivial result is a non-adaptive algorithm that uses $\bigo\left(
        \frac{d}{\varepsilon^2} + \frac{d^2}{\varepsilon}
    \right)$ copies~\cite{pelecanos2026beating}, which is polynomially larger than the best-known lower bound. Our first result is a new algorithm performing non-adaptive incoherent measurements that succeeds using $\bigo\left(
        \frac{d}{\varepsilon^2} + \frac{d^{3/2}}{\varepsilon}
    \right)$ copies, improving on the latter term in the copy complexity. Our estimator performs random basis measurements on individual copies in batches, and estimates the collision probability of the output distributions. Moreover, we show that for any algorithm restricted to non-adaptive measurements, the above copy complexity is optimal.

    We also develop a new \emph{adaptive} estimator for the purity of a state that improves on the above complexity in the high-precision regime, i.e., for $\eps = o(1/d)$. Our estimator proceeds in two stages: first, we use an initial batch of copies to estimate the state, and then measure the remaining copies in the estimated basis. We also show that the non-adaptive and adaptive estimators together yield the optimal complexity for incoherent purity estimation; in particular, we show that the copy complexity of this problem is $$
        \Theta\left(\min\left\{
        \frac{d}{\varepsilon^2} + \frac{d^{3/2}}{\varepsilon}, \frac{d^2}{\varepsilon} + \frac{\sqrt{d}}{\varepsilon^2}\right\}
    \right). 
    $$ 
    We prove our adaptive lower bounds by analyzing a suitable mixture-versus-mixture testing problem. We achieve this by showing that a classical problem of testing Gaussian distributions from noisy queries reduces to this quantum testing problem with a constant-factor overhead and then proving the desired lower bound for the classical task. We expect such reductions to have further applications in proving quantum inference lower bounds against adaptive algorithms. 
\end{abstract}

\newpage

\tableofcontents

\newpage

\input{sections/1_intro.tex}
\input{sections/2_preliminaries}
\input{sections/3_nonadaptive_estimator}

\input{sections/4_adaptive_estimator}

\input{sections/5_nonadaptive_lower_bounds}
\input{sections/6_adaptive_lower}

\appendix
\input{sections/A_gaussian_simulation}

\bibliographystyle{alpha}
\bibliography{biblio}

\end{document}

%% file: header/usepackages.tex
\usepackage[utf8]{inputenc}

\usepackage{xcolor}
\definecolor{blueviolet}{rgb}{0.2, 0.2, 0.6}
\definecolor{webgreen}{rgb}{0,.5,0}
\definecolor{webbrown}{rgb}{.6,0,0}
\usepackage[pdftex,
  bookmarks=false,
  colorlinks=true, %allcolors=blueviolet,
  urlcolor=webbrown,
  linkcolor=blueviolet, 
  citecolor=webgreen,
  pdfstartpage=1,
  pdfstartview={FitH},  % FitBH
  bookmarksopen=false,
  pagebackref
  ]{hyperref}
\renewcommand*{\backref}[1]{}
\renewcommand*{\backrefalt}[4]{%
    \ifcase #1%
          \or Cited on page~#2.%
          \else Cited on pages~#2.%
    \fi%
    }
\usepackage{amsmath}
\usepackage{amssymb}
\usepackage{amsfonts}
\usepackage{amsthm}
\usepackage{zref-clever}
\zcsetup{cap = true}
\newcommand{\Cref}{\zcref}
\allowdisplaybreaks %Toggle page breaks in align environments

\usepackage{enumerate}
\usepackage{enumitem}
\usepackage{comment}
\usepackage{xpatch}
\usepackage{graphicx}
\usepackage{tabularx}
\usepackage{braket}
\usepackage{amsthm}
\usepackage{tikz}
\usepackage{commath}
\usepackage{mathtools}
\usepackage{qcircuit}
\usepackage{algorithm}
\usepackage{algpseudocodex}[indLines = true,italicComments = false, ]
\usepackage{tabto}
\usepackage{pgf-umlsd}
\usepackage{bm}
\usepackage{bbm}
\usepackage{multirow, multicol}
\usepackage{float}
\usepackage{pdfpages}
\usepackage{caption}
\usepackage{subcaption}
\usepackage{makecell}
\usepackage{cite}
\usepackage{tablefootnote}

\usepackage{pgfplots}
\pgfplotsset{compat=1.18}
\usetikzlibrary{arrows.meta}
\usepackage{xcolor}

\definecolor{adaptiveblue}{HTML}{6FA8FF}
\definecolor{plum}{HTML}{6D4A7F}
\definecolor{phasegray}{HTML}{D5D8DE}
\definecolor{axisgray}{HTML}{3A3A3A}

%% file: header/commands.tex
\let\Pr\relax
\DeclareMathOperator*{\Pr}{\mathbf{Pr}}
\DeclareMathOperator*{\Var}{\mathbf{Var}}

\let\originalleft\left
\let\originalright\right
\renewcommand{\left}{\mathopen{}\mathclose\bgroup\originalleft}
\renewcommand{\right}{\aftergroup\egroup\originalright}

\newcommand{\mc}[1]{\mathcal{#1}}
\newcommand{\mbb}[1]{\mathbb{#1}}

\newcommand{\mrm}[1]{\mathrm{#1}}

\newcommand{\dtv}{\mathrm{d_{TV}}}
\newcommand{\dkl}{\mathrm{KL}}

\newcommand{\bigo}{\mathcal{O}}

\newcommand{\bfY}{\boldsymbol{Y}}
\newcommand{\bfZ}{\boldsymbol{Z}}
\newcommand{\bftheta}{\boldsymbol{\theta}}

\newcommand{\eps}{\varepsilon}

\newcommand{\tr}{\mathrm{tr}}

\newcommand{\swap}{\mathrm{SWAP}}

\newtheorem{theorem}{Theorem}[section]
\newtheorem{proposition}[theorem]{Proposition}
\newtheorem{lemma}[theorem]{Lemma}

\newtheorem{fact}[theorem]{Fact}

\AddToHook{env/theorem/begin}{%
\zcsetup{countertype={theorem=theorem}}}
\AddToHook{env/proposition/begin}{%
\zcsetup{countertype={theorem=proposition}}}
\AddToHook{env/lemma/begin}{%
\zcsetup{countertype={theorem=lemma}}}
\AddToHook{env/corollary/begin}{%
\zcsetup{countertype={theorem=corollary}}}
\AddToHook{env/claim/begin}{%
\zcsetup{countertype={theorem=claim}}}
\AddToHook{env/problem/begin}{%
\zcsetup{countertype={theorem=problem}}}
\AddToHook{env/conjecture/begin}{%
\zcsetup{countertype={theorem=conjecture}}}
\AddToHook{env/fact/begin}{%
\zcsetup{countertype={theorem=fact}}}
\zcRefTypeSetup{fact}{
Name-sg = Fact ,
name-sg = fact ,
Name-pl = Facts ,
name-pl = facts ,
}
\AddToHook{env/question/begin}{%
\zcsetup{countertype={theorem=question}}}
\AddToHook{env/claim/begin}{%
\zcsetup{countertype={theorem=claim}}}
\zcRefTypeSetup{claim}{
Name-sg = Claim ,
name-sg = claim ,
Name-pl = Claims ,
name-pl = claims ,
}

\theoremstyle{definition}
\newtheorem{definition}[theorem]{Definition}

\AddToHook{env/definition/begin}{%
\zcsetup{countertype={theorem=definition}}}
\AddToHook{env/notation/begin}{%
\zcsetup{countertype={theorem=notation}}}
\AddToHook{env/example/begin}{%
\zcsetup{countertype={theorem=example}}}
\AddToHook{env/assumption/begin}{%
\zcsetup{countertype={theorem=assumption}}}
\AddToHook{env/remark/begin}{%
\zcsetup{countertype={theorem=remark}}}

\numberwithin{equation}{section}

%% file: header/comments.tex
\definecolor{chiragcolor}{RGB}{66,175,210}
\definecolor{junseocolor}{rgb}{0, 0.125, 0.376} %Feel free to pick a different color here @Junseo

\makeatletter
\newif\if@comments
\newcommand{\chirag}[1]{\if@comments\textcolor{chiragcolor}{[CW: #1]}\fi}
\newcommand{\junseo}[1]{\if@comments\textcolor{junseocolor}{[JL: #1]}\fi}
\newcommand\enablecomments{\@commentstrue}
\newcommand\disablecomments{\@commentsfalse}
\enablecomments

\makeatother

%% file: sections/1_intro.tex
\section{Introduction}
\label{sec:intro}

The purity $\tr(\rho^2)$ is one of the simplest and most fundamental properties of a quantum state $\rho$. It quantifies mixedness, ranging from $1/d$ for the maximally mixed state to $1$ for a pure state, and determines the second R\'enyi entropy $S_2(\rho)=-\log\tr(\rho^2)$. For a subsystem of a pure bipartite state, this entropy measures entanglement with the remaining system. Purity estimation is thus a fundamental primitive for characterizing quantum states and probing many-body entanglement, with experimental applications based on randomized measurements~\cite{BEJ+19,EVRZ19}. Although purity can be measured directly through the identity $\tr(\rho^2)=\tr(\operatorname{SWAP}\rho^{\otimes 2})$, this approach requires coherent access to two copies of the state. Such access can be experimentally demanding, motivating protocols that measure each copy separately and combine only classical measurement outcomes. We study this experimentally relevant model of \emph{incoherent} or \emph{single-copy} measurements: arbitrary measurements on each copy are allowed, but no quantum information is retained across copies.

Despite the fundamental nature of purity estimation, its optimal sample complexity with incoherent measurements has remained unresolved. Randomized measurements and classical shadows provide general approaches to estimating purity and related nonlinear properties~\cite{EVRZ19,HKP20}. For additive error $\eta$, the inner-product estimator of Anshu, Landau, and Liu~\cite{ALL22} yields a non-adaptive purity estimator using $\bigo(\max\{1/\eta^2,\sqrt{d}/\eta\})$ copies. Chen, Cotler, Huang, and Li~\cite{CCHL21} established an $\Omega(\sqrt{d})$ lower bound even for constant additive error and adaptive single-copy measurements. Gong, Haferkamp, Ye, and Zhang~\cite{GHYZ24} subsequently proved the precision-dependent lower bound $\Omega(\max\{1/\eta^2,\sqrt{d}/\sqrt{\eta}\})$, and obtained the stronger bound $\Omega(\max\{1/\eta^2,\sqrt{d}/\eta\})$ for protocols that repeat an identical single-copy projective measurement. Here we focus on \emph{relative} error $\eps$, which remains informative even when the purity is as small as $1/d$ and, for $0<\eps\le 1/2$, corresponds to additive $\bigo(\eps)$ accuracy in the second R\'enyi entropy. With unrestricted collective measurements, the optimal complexity is $\Theta(\max\{\sqrt{d}/\eps^2,d/\eps\})$~\cite{AISW20}. With incoherent measurements, however, the best previously known upper bound was $\bigo(\max\{d/\varepsilon^2,d^2/\varepsilon\})$, obtained by specializing the moment-estimation bounds of Pelecanos, Tan, Tang, and Wright~\cite{pelecanos2026beating} to the second moment; the same bound is recovered by the single-copy specialization of the limited-entanglement estimator of Wadhwa and Chen~\cite{WC26}. Meanwhile, the mixedness-testing lower bound of~\cite{chen2022tight}, together with the collective-measurement lower bound of~\cite{AISW20}, only implies $\Omega(\max\{d^{3/2}/\varepsilon, \sqrt{d}/\varepsilon^2\})$ copies are necessary, even when measurements may be chosen adaptively. These bounds leave a polynomial gap and do not determine the optimal dependence on dimension and precision. Our first goal is to close this gap and determine the number of copies needed to estimate purity as a function of dimension and precision. First, we formalize the relative-error estimation task.

\begin{definition}[Purity estimation]
\label{def:purity-estimation}
    We say that an algorithm succeeds at $(d,\varepsilon)$-purity estimation if, given copies of any state $\rho$, it produces an estimate $\hat{P}$ such that\footnote{The success probability $\frac23$ here is arbitrary and can be boosted to $1-\delta$ for any $\delta > 0$ at the cost of $\bigo(\log(1/\delta))$ repetitions.}
    \begin{equation}
        \mathbf{Pr}[|\hat{P} - \tr(\rho^2)| \leq \eps \cdot \tr(\rho^2)] \geq \frac23.
    \end{equation}
\end{definition}

We will follow the above definition throughout the paper, and may omit ``$(d,\varepsilon)$'' when clear from context. With this definition in place, our central question is the following.

\begin{quote}
    \textbf{Question 1.} \emph{What is the optimal sample complexity of $(d,\eps)$-purity estimation using incoherent measurements?}
\end{quote}

We also note that the single-copy estimators achieving the upper bounds discussed above are non-adaptive, i.e., all measurements are fixed before any outcomes are observed. An adaptive incoherent protocol can instead use previous outcomes to choose subsequent measurements, potentially directing its measurements toward features of the unknown state that are most informative about its purity. Whether this flexibility reduces copy complexity depends on the task. For example, adaptivity does not improve the asymptotic complexity of state certification~\cite{chen2022tight}, but it yields a near-quadratic improvement in the precision dependence of quantum state tomography under infidelity loss~\cite{CHLLS23}. Since purity depends only on the spectrum of the state, it is natural to ask whether learning from earlier outcomes can nevertheless help estimate this single scalar quantity.

\begin{quote}
    \textbf{Question 2.} \emph{Is there any regime of $d,\varepsilon$ in which adaptivity helps reduce the copy complexity of purity estimation with single-copy measurements?}
\end{quote}

\subsection{Our results}

As mentioned previously, the best-known upper bound for purity estimation with incoherent measurements is $\bigo(\max\{d/\varepsilon^2, d^2/\varepsilon\})$, implied by the moment estimation bounds of \cite{pelecanos2026beating}. Our first result is a non-adaptive algorithm that improves on this bound, which we also show to be optimal.
\begin{theorem}[Purity estimation with non-adaptive measurements]
    \label{thm:intro-non-adaptive}
    For all sufficiently large $d$ and sufficiently small $\eps$, the copy complexity of $(d,\varepsilon)$-purity estimation with \emph{non-adaptive} incoherent measurements is 
    \begin{equation}
        \Theta\Bigl(
            \max \Bigl\{
            \frac{d}{\varepsilon^2}, \,\frac{d^{3/2}}{\varepsilon}
            \Bigr\}
        \Bigr).
    \end{equation}
\end{theorem}

The above theorem partially answers our first central question by pinning down the copy complexity with \emph{non-adaptive} measurements; addressing it fully requires extending our results to the setting where measurements can be chosen \emph{adaptively}.

In this setting of adaptive measurements, we actually show that the above complexity is \emph{not} tight: we develop a new adaptive estimator that improves on this complexity in certain parameter regimes. In particular, our adaptive estimator has complexity $\bigo(\max\{d^2/\varepsilon, \sqrt{d}/\varepsilon^2\})$. As we demonstrate in \Cref{fig:main}, this improves on the non-adaptive estimator whenever $\eps \lesssim 1/d$. This answers our second question in the affirmative, showing that adaptivity indeed helps for purity estimation in the high-precision regime. We also prove adaptive lower bounds that match both algorithms in their respective parameter regimes, obtaining an optimal characterization of purity estimation with single-copy measurements and fully addressing our first question as well.

\begin{theorem}[Purity estimation with adaptive measurements]
\label{thm:intro-adaptive}
   For all sufficiently large $d$ and sufficiently small $\eps$, the copy complexity of $(d,\varepsilon)$-purity estimation with incoherent \emph{adaptively chosen} measurements is
    \begin{equation}
        \Theta\Bigl(
            \min \Bigl\{
                \frac{d}{\eps^2} + \frac{d^{3/2}}{\eps}, \, \frac{d^2}{\eps} + \frac{\sqrt{d}}{\eps^2} 
            \Bigr\}
        \Bigr).
    \end{equation}
\end{theorem}

We demonstrate how the above complexity transitions across parameter regimes in \Cref{fig:main}. Interestingly, as we move across parameter regimes, the complexity alternates between Heisenberg ($\eps^{-1}$) and standard quantum scaling ($\eps^{-2}$). Together, our results show that adaptivity is indeed necessary for optimal purity estimation; however, this is only the case in the high-precision regime, i.e., for $\varepsilon = o(1/d)$. 

\input{sections/figure}

\subsection{Technical overview}

We provide an overview of our two purity estimators in \Cref{sec:overview-estimators}, and present the lower bound ideas in \Cref{sec:overview-lower}. Throughout this overview, for integer $k \geq 2$, we denote $P_k \coloneqq \operatorname{tr}(\rho^k)$, and follow \Cref{def:purity-estimation}.

\subsubsection{Our estimators}
\label{sec:overview-estimators}

We develop two complementary purity estimation algorithms from single-copy measurements. Our non-adaptive estimator applies random basis measurements and estimates the collision probability of the outcome distribution. Our adaptive estimator instead first obtains a preliminary estimate of the unknown state, and uses this to refine further measurements.

\paragraph{Non-adaptive estimator:}
Given $n$ copies of the state, we split them up into $B$ batches of $s = n/B$ copies each. For each batch, we pick a Haar-random basis and measure all copies in the batch in this basis. For the $b$th batch, let $N_{b,i}$ be the number of occurrences of outcome $i$. We use these outcomes to construct an unbiased estimate of this basis's collision probability:
\begin{equation}
    \widehat{Q}_b \coloneqq \frac{\sum_i N_{b,i} (N_{b,i}-1)}{s(s-1)}.
\end{equation}
Averaging over the Haar measure, we show that
\begin{equation}
    \mathbb{E}[\widehat{Q}_b] = \frac{1 + P_2}{d+1}.
\end{equation}
We debias this estimator and average over all batches to obtain our final estimator:

\begin{equation}
    \widehat{P}_2 = \frac{1}{B} \sum_{b = 1}^B ((d+1) \widehat{Q}_b - 1).
\end{equation}
We analyze the variance of this estimator in \Cref{na-upper:variance-bound} and balance $B,s$ to keep the variance below $\bigo(\eps^2 P_2^2)$, ensuring that our estimator is accurate with high probability. Specifically, we take $B=\Theta(\max\{1,1/(d\varepsilon^2)\})$ and $s=\Theta(\min\{d^{3/2}/\varepsilon,d^2\})$, which yields the claimed upper bound of $\bigo(d^{3/2}/\varepsilon+d/\varepsilon^2)$ in \Cref{thm:intro-non-adaptive}.

\paragraph{Adaptive estimator:} This estimator proceeds in two stages. First, we take $n$ copies of the unknown state $\rho$, and apply a single-copy state estimator to each one, averaging the results to get an unbiased estimate $\hat{\rho}$. We then measure $m$ copies of $\rho$ individually in the basis of $\hat{\rho}$ to estimate the overlap $\tr(\rho \hat{\rho})$; let this final estimate be $\overline{Y}$. As $\overline{Y}$ is an unbiased estimator for the purity, a natural strategy would be to immediately output this estimate. However, the variance of $\overline{Y}$ contains a prohibitively large error term from the state estimation routine, which prevents one from improving over the non-adaptive rate.

Instead, we propose a new unbiased estimator that corrects for the state estimation error and reduces the variance. In particular, we compute
\begin{equation}
    \hat{P}_2 = \frac{2\overline{Y} - \tr(\hat{\rho}^2) + D/n}{1+1/n}, \quad \textnormal{where } D = d^2 + d - 1.
\end{equation}
The key reason for this variance reduction is that, conditioned on $\hat{\rho}$, the first two terms in the numerator have expectation $\tr(\rho^2) - \|\rho - \hat{\rho}\|_2^2$. Over the randomness of $\hat{\rho}$, the variance of this quadratic error term ends up being smaller than that of $\tr(\rho\hat{\rho})$ in relevant regimes, allowing us to obtain our improved adaptive complexity.

\subsubsection{Lower bound ideas}
\label{sec:overview-lower}

\paragraph{Non-adaptive lower bounds:}
The two terms in the non-adaptive lower bound come from different testing problems. First, the identity $P_2=1/d+\|\rho-I/d\|_F^2$ implies that relative-error purity estimation solves mixedness testing at trace-norm separation $\Theta(\sqrt\varepsilon)$. The lower bound of~\cite{chen2022tight} therefore gives $\Omega(d^{3/2}/\varepsilon)$ copies, even for adaptive measurements.

For the $\Omega(d/\varepsilon^2)$ term, we consider states $\rho_{\theta,v}=(1-\theta)I/d+\theta|v\rangle\langle v|$, where $v$ is a Haar-random direction drawn once and held fixed across all copies. We compare spike strengths $\theta_0=d^{-1/2}$ and $\theta_1=(1+8\varepsilon)d^{-1/2}$. Their purities, given by $1/d+(1-1/d)\theta^2$, have disjoint relative-error intervals. For any POVM chosen independently of $v$, Haar second-moment identities bound the average relative entropy between its two outcome distributions by $\bigo(\varepsilon^2/d)$. Non-adaptivity allows us to sum this bound over all copies, even for different POVMs and shared measurement randomness. Distinguishing the two hypotheses with constant advantage consequently requires $\Omega(d/\varepsilon^2)$ copies.

\paragraph{Prior adaptive lower bounds:} As stated previously, the mixedness testing lower bound of \cite{chen2022tight} already implies the $\Omega(d^{3/2}/\varepsilon)$ term. Moreover, \cite{AISW20} already showed an $\Omega(\sqrt{d}/\eps^2)$ lower bound for purity estimation even with entangled measurements. Our main technical contribution here is an $\Omega\Bigl( \min\Bigl\{
    \frac{d}{\eps^2}, \frac{d^2}{\eps}
\Bigr\}\Bigr)$ lower bound. We note that these three lower bounds together are equivalent to the lower bound of \Cref{thm:intro-adaptive}, and so it suffices for us to prove this latter bound; the rest of the section outlines this proof.

\paragraph{Gaussian encoding states:} Our hard instance will be two mixtures of states with separated purities. In particular, let $V_1, \dots, V_{d^2-1}$ be an orthonormal basis of the space of Hermitian traceless $d \times d$ matrices. We will pick a random vector $\bftheta \sim \mc{N}(0, v I_{d^2-1})$, and use these as coefficients to perturb the maximally mixed state. Specifically, we define
\begin{equation}
    \Delta_\theta \coloneqq \sum_{r = 1}^{d^2-1} \theta_r V_r, \quad \rho_\theta \coloneqq \frac{I}{d} + \Delta_\theta.
\end{equation}
We note that similar hard instances have previously been used for mixedness testing lower bounds in various settings~\cite{liu2024role,liu2024quantum,chen2026distributed,WC26}, where the perturbation coefficients were Rademacher random variables instead of Gaussian ones. While one might expect such sub-Gaussian vectors to suffice for our lower bound proofs, as we explain later, it will be crucial to our analysis that the coefficients are \emph{exactly} Gaussian, necessitating this new hard instance. 

In \Cref{lem:instance-valid-state}, we show that the operators $\rho_{\bftheta}$ are highly likely to be valid quantum states. We will select two variance parameters $v_1$ and $v_0$, and use these to parameterize the perturbation vector $\bftheta \sim \mc{N}(0, v_i I_{d^2-1})$; the tester will be tasked with distinguishing between these two state mixtures. In expectation, the resulting states have purity $\frac1d + v_i(d^2-1)$, and in \Cref{lem:instance-purity-separation} we show that the states have purity sufficiently concentrated around their mean. Consequently, for $v_0$ and $v_1$ chosen appropriately, a purity estimation algorithm can distinguish between the two mixtures, making this a suitable hard instance.

\paragraph{Reduction to Gaussian testing:} Unfortunately, we cannot directly prove the indistinguishability of these two mixtures. While the operators are \emph{likely} to be valid quantum states, not every possible operator $\rho_\theta$ is a valid quantum state. While we could get around this by conditioning only on the valid vectors, the conditional distribution of these vectors would no longer be Gaussian, removing our ability to exploit several desirable properties of such distributions. Instead, we develop a reduction to a classical testing problem for Gaussian distributions, and embed the undesirable events in low-probability branches of a classical simulator.

We will still aim to prove a lower bound for determining $i$ given access to $\bftheta \sim \mc{N}(0, v_i I_{d^2-1})$. However, instead of access to the encoding quantum states $\rho_{\bftheta}$, we consider a noisy query model: an algorithm can provide an input vector $a \in \mathbb{R}^{d^2-1}$ with bounded norm, and receives a random output
\begin{equation}
    Y = a^\top \bftheta + \bfZ, \quad \bfZ \sim \mc{N}(0,1).
\end{equation}
In \Cref{lem:gaussian-simulation}, we show that any single-copy POVM on a state $\rho_\theta$ can be simulated \emph{exactly} with a constant expected number of noisy queries to an unknown vector $\theta$. Consequently, one can take a purity estimation algorithm and simulate all its measurements with these noisy linear queries to the vector $\theta$, giving us a classical algorithm for the Gaussian testing problem. Conversely, a lower bound for this classical problem implies the same (up to constant factors) for purity estimation.

Note that the above reduction helps us get around the aforementioned issue of certain vectors $\theta$ not corresponding to valid states. Even in such cases, the classical testing problem remains well-defined, obviating the need to study conditional priors. Moreover, we can allow the simulated purity estimator to act arbitrarily on these low-probability events and simply adjust the overall success probability of the classical tester. We handle this formally in the proof of \Cref{thm:adaptive-lower-main}. We highlight that this reduction also allows us to map a quantum mixture-versus-mixture distinguishing task to that of distinguishing between two fixed Gaussian distributions, greatly simplifying the subsequent analysis.

\paragraph{Gaussian testing lower bound:} We briefly discuss our classical lower bound techniques, with full details in \Cref{sec:gaussian-lower}. Prior classical work has considered related adaptive inference problems from noisy queries, such as estimating an unknown deterministic vector $\theta$~\cite{arias2012fundamental} and testing or detection problems for sparse vectors~\cite{castro2014adaptive}. However, to our knowledge, testing the scale of a Gaussian vector from adaptive noisy queries has not been studied previously. Our classical lower bound may therefore be of independent interest.

We prove this lower bound by showing that the distributions over query transcripts are indistinguishable across the two regimes. To this end, we relate the total variation distance between these distributions to the Fisher information of a transcript with respect to the parameter $v_i$. We directly analyze this quantity by proving uniform upper bounds on the per-query increment to the Fisher information. A key property we exploit for this purpose is that the posterior distribution of the vector $\theta$ \emph{remains Gaussian} after a noisy linear query, and the parameter updates can be tracked explicitly (see~\Cref{lem:posterior-updates}). This property is another crucial simplification offered by our use of the Gaussian-encoded hard instance over the standard Rademacher one.

\subsection{Related work}

\paragraph{Purity and moment estimation.}
Purity estimation belongs to the broader problem of estimating spectral moments and R\'enyi entropies without reconstructing an entire state. Early work of Ekert et al.~\cite{E02} introduced controlled-SWAP networks for directly estimating nonlinear state functionals, while Brun~\cite{Brun04} showed how polynomial functionals can be measured using observables on multiple copies. These questions have close classical counterparts: the power sum $\sum_{i=1}^d p_i^k$ of a probability distribution can be estimated from $k$-way sample collisions. For each fixed integer $k\ge 2$, Acharya, Orlitsky, Suresh, and Tyagi~\cite{AOST17} established the optimal sample complexity $\Theta(d^{1-1/k}/\varepsilon^2)$ for relative-error estimation of these moments or equivalently additive-error estimation of the corresponding R\'enyi entropy. In particular, the classical collision probability requires $\Theta(\sqrt d/\varepsilon^2)$ samples. Acharya, Issa, Shende, and Wagner~\cite{AISW20} characterized the analogous quantum problem with unrestricted collective measurements, obtaining $\Theta(\max\{d^{1-1/k}/\varepsilon^2,d^{2-2/k}/\varepsilon^{2/k}\})$ copies. For purity, their bound specializes to $\Theta(\max\{\sqrt d/\varepsilon^2,d/\varepsilon\})$. Measuring in a known eigenbasis reduces the problem to classical sampling; our setting concerns an unknown eigenbasis and measurements restricted to individual copies.

\paragraph{Incoherent estimation.}
Randomized measurements provide a general route to estimating nonlinear properties from single-copy data. Van Enk and Beenakker~\cite{vEB12} studied spectral-moment estimation using random measurements, and subsequent protocols based on correlations between locally randomized measurements enabled purity and entanglement-entropy estimation in many-body systems~\cite{EVRZ19,BEJ+19}. Classical shadows offer a related framework for estimating linear and nonlinear state properties from reusable classical measurement records~\cite{HKP20}. Here, the single-copy restriction permits arbitrary measurements within each copy and therefore differs from a restriction to spatially local measurements. For additive purity error $\eta$, the distributed inner-product estimator of Anshu, Landau, and Liu~\cite{ALL22} implies a non-adaptive upper bound of $\bigo(\max\{\eta^{-2},\sqrt d/\eta\})$. Their matching lower bound concerns the distributed inner-product problem and does not directly imply a matching purity lower bound. For purity itself, Chen, Cotler, Huang, and Li~\cite{CCHL21} established an $\Omega(\sqrt d)$ lower bound at constant additive accuracy, even for adaptive single-copy measurements. Gong, Haferkamp, Ye, and Zhang~\cite{GHYZ24} strengthened this to $\Omega(\max\{\eta^{-2},\sqrt{d/\eta}\})$ and proved the stronger bound $\Omega(\max\{\eta^{-2},\sqrt d/\eta\})$ for protocols repeating an identical single-copy projective measurement. The latter restriction is more restrictive than non-adaptivity and does not allow for arbitrary non-adaptive measurement schedules. More recently, Akresh and Beckey~\cite{AB26} obtained sharp constant-accuracy purity-testing lower bounds using a positive-partial-transpose relaxation that includes adaptive single-copy protocols. For relative-error estimation, Pelecanos, Tan, Tang, and Wright~\cite{pelecanos2026beating} developed single-copy moment estimators as a component of their spectrum-learning algorithm; the second-moment specialization uses $\bigo(\max\{d/\varepsilon^2,d^2/\varepsilon\})$ copies. The same purity bound follows from the single-copy specialization of the limited-entanglement estimator of Wadhwa and Chen~\cite{WC26}. Our non-adaptive result improves the $d^2/\varepsilon$ term to $d^{3/2}/\varepsilon$ and establishes the optimal rate over all non-adaptive single-copy protocols.

\paragraph{Power of adaptivity.}
The value of adapting measurements to earlier outcomes depends on both the task and the loss function. Adaptive purity estimation already appeared in the qubit setting: Bagan, Ballester, Mu\~noz-Tapia, and Romero-Isart~\cite{BBMR05} gave a two-stage separable-measurement protocol that asymptotically attains the optimal collective-measurement performance under a Bayesian fidelity criterion. Their protocol first estimates the Bloch direction and then measures the remaining copies along the estimated direction. This fixed-dimensional asymptotic result differs from the uniform, relative-error guarantees studied here. In higher dimensions, Chen, Huang, Li, and Liu~\cite{chen2022tight} showed that adaptivity does not improve the asymptotic complexity of mixedness testing: $\Theta(d^{3/2}/\alpha^2)$ copies are necessary and sufficient to distinguish $I/d$ from states at trace distance at least $\alpha$. Chen, Huang, Li, Liu, and Sellke~\cite{CHLLS23} established the same absence of an asymptotic advantage for tomography in trace distance, but showed that this resource improves infidelity tomography from a non-adaptive rate of $\Omega(d^3/\gamma^2)$ to an adaptive one of $\widetilde{\bigo}(d^3/\gamma)$ for target infidelity $\gamma$. Our results identify a precision-dependent advantage for estimating a single spectral functional. Together with the matching lower bounds, our results show that adaptivity improves the asymptotic rate when $\varepsilon=o(1/d)$, with the advantage reaching a factor of order $\sqrt d$ when $\varepsilon\lesssim d^{-3/2}$. Thus, the benefit of adaptivity for purity estimation is governed jointly by dimension and relative precision.

\subsection*{Acknowledgments}
The authors thank Qisheng Wang and Angus Lowe for helpful discussions at an early stage of this project. J.L. acknowledges support from the Harvard Quantum Initiative and the Kwanjeong Foundation, and this work benefited from interactions with the NSF AI Institute for Artificial Intelligence and Fundamental Interactions (IAIFI), which is supported by the National Science Foundation under Cooperative Agreement PHY-2525568. C.W. was supported by the UK EPSRC through the Quantum Advantage Pathfinder project with grant
reference EP/X026167/1.

\subsection*{LLM Use}

The proof ideas for our main results were developed by frontier LLMs, except for the non-adaptive estimator where LLMs only assisted with calculations. The authors digested these proofs, verified all steps, and carried out substantive rewrites, and are fully responsible for the paper's contents. An LLM also assisted with the TikZ code used to generate \Cref{fig:main}.

%% file: sections/figure.tex
\begin{figure}[t]
\centering
\begin{tikzpicture}
\begin{axis}[
    width=0.92\linewidth,
    height=0.42\linewidth,
    xmin=0, xmax=2.0,
    ymin=0, ymax=34,
    axis lines=none,
    clip=false,
    enlargelimits=false,
    xtick=\empty,
    ytick=\empty,
    tick label style={font=\small},
]

% --- epsilon-regime labels below the x-axis ---

\node[
    anchor=north,
    font=\small
] at (axis description cs:0.12,0)
    {$\epsilon \gtrsim d^{-1/2}$};

\node[
    anchor=north,
    font=\small
] at (axis description cs:0.375,0)
    {$d^{-1}\!\lesssim\!\epsilon\!\lesssim\! d^{-1/2}$};

\node[
    anchor=north,
    font=\small
] at (axis description cs:0.625,0)
    {$d^{-3/2}\!\lesssim\!\epsilon\!\lesssim\! d^{-1}$};

\node[
    anchor=north,
    font=\small
] at (axis description cs:0.875,0)
    {$\epsilon \lesssim d^{-3/2}$};

% y-axis label
\node[
    anchor=south,
    font=\small
] at (axis description cs:0,1.02)
    {Copy complexity};

\node[
    anchor=north east,
    font=\small
] at (axis description cs:1.1,0)
    {$\log(1/\epsilon)$};

% --- vertical phase boundaries ---
\addplot[
    densely dotted,
    line width=1pt,
    color=phasegray
] coordinates {(0.5,0) (0.5,34)};
\addplot[
    densely dotted,
    line width=1pt,
    color=phasegray
] coordinates {(1.0,0) (1.0,34)};
\addplot[
    densely dotted,
    line width=1pt,
    color=phasegray
] coordinates {(1.5,0) (1.5,34)};

% --- adaptive curve: 2^{min(max(3/2+\alpha,1+2\alpha), max(2+\alpha,1/2+2\alpha))} ---
\addplot[
    domain=0:2,
    samples=300,
    smooth,
    line width=1.1pt,
    color=adaptiveblue
]
{pow(2, min(max(1.5 + x, 1 + 2*x), max(2 + x, 0.5 + 2*x)))};

% --- non-adaptive curve: 2^{max(3/2+\alpha,1+2\alpha)} ---
\addplot[
    domain=0:2,
    samples=300,
    smooth,
    line width=1.0pt,
    dashed,
    dash pattern=on 5pt off 3pt,
    color=plum
]
{pow(2, max(1.5 + x, 1 + 2*x))};

% --- axes with aligned arrowheads ---
\draw[-{Latex[length=2mm,width=1.4mm]}, line width=0.45pt, color=axisgray]
    (rel axis cs:0,0) -- (rel axis cs:1.03,0);

\draw[-{Latex[length=2mm,width=1.4mm]}, line width=0.45pt, color=axisgray]
    (rel axis cs:0,0) -- (rel axis cs:0,1.03);

% --- labels for active complexity regimes ---
\node[anchor=south] at (axis cs:0.22,4.2) {$\Theta\bigl(\frac{d^{3/2}}{\epsilon}\bigr)$};
\node[anchor=south] at (axis cs:0.73,6.0) {$\Theta\bigl(\frac{d}{\epsilon^2}\bigr)$};
\node[anchor=north] at (axis cs:1.22,9.0) {$\Theta\bigl(\frac{d^2}{\epsilon}\bigr)$};
\node[anchor=north] at (axis cs:1.77,15.8) {$\Theta\bigl(\frac{\sqrt{d}}{\epsilon^2}\bigr)$};

% --- curve labels outside the plotting region ---
\node[
    anchor=west,
    text=plum,
    font=\small\sffamily
] at (axis cs:2.04,32.0) {Non-adaptive};

\node[
    anchor=west,
    text=adaptiveblue,
    font=\small\sffamily
] at (axis cs:2.04,22.6) {Adaptive};

\end{axis}
\end{tikzpicture}

\caption{Demonstration of the copy complexity across the different regimes of $\epsilon$. The {\color{plum} non-adaptive} curve corresponds to the complexity of \Cref{thm:intro-non-adaptive}, while the {\color{adaptiveblue} adaptive} one corresponds to \Cref{thm:intro-adaptive}. For $\eps \gtrsim d^{-1}$, the non-adaptive estimator already attains the optimal adaptive complexity, making the curves coincide. For $\eps \ll d^{-1}$, the adaptive estimator instead outperforms the non-adaptive one, demonstrating a strict separation.}
\label{fig:main}
\end{figure}
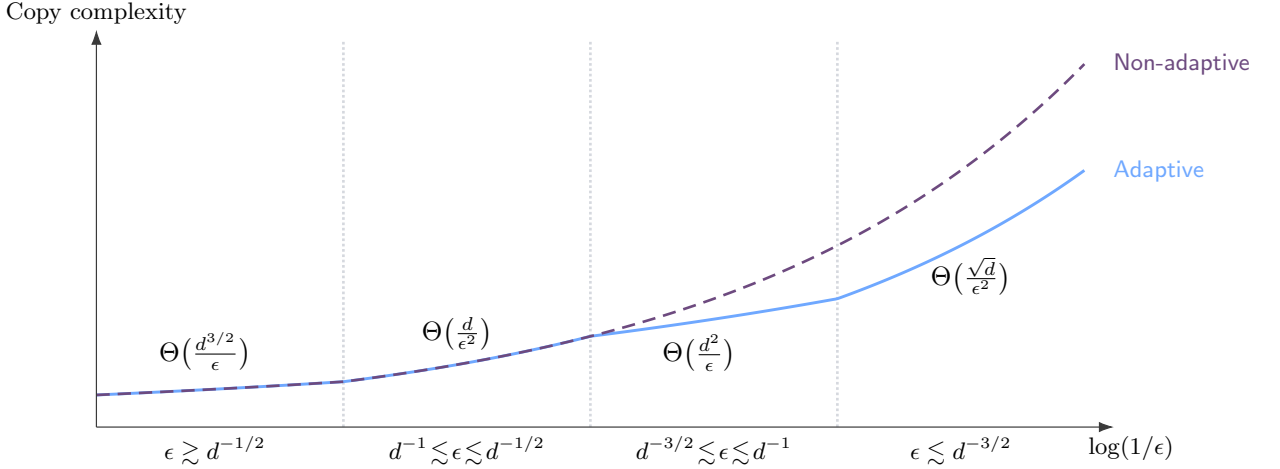

%% file: sections/2_preliminaries.tex
\section{Preliminaries}
\label{sec:prelim}

Throughout, for integer $k \geq 1$, we will use $P_k(\rho) \coloneqq \tr(\rho^k)$ to denote the $k$th spectral power sum of a state $\rho$. When $\rho$ is clear from context, we may omit this and only use $P_k$. Throughout, for $p \geq 1$, $\|\cdot\|_p$ denotes the $\ell_p$-norm for vectors and the Schatten $p$-norm for operators.

In this paper, we only consider algorithms performing single-copy measurements. When all POVMs are chosen \emph{before} any outcomes are observed, we call the ordered list of measurements a \emph{non-adaptive} POVM schedule, and say the algorithm is non-adaptive. Note that non-adaptive algorithms may use shared randomness to determine the measurement schedule, as long as this is done before any measurements are actually performed. More generally, algorithms can also choose measurements \emph{adaptively} based on prior outcomes. However, conditioning on any shared randomness and prior measurement outcomes, this choice is deterministic. One can thus model an adaptive schedule of measurements as a deterministic decision tree, where the branches are chosen according to the current measurement outcome and all pre-shared randomness.

To prove our lower bounds, we will show that output distributions under various measurement schedules are indistinguishable. To this end, we will need various notions of distances between distributions.

\begin{definition}[TV distance, KL divergence]
    Let $P,Q$ be two discrete distributions. Their TV distance and KL divergences are given by
    \begin{equation}
        \dtv(P,Q) = \frac12 \sum_x |P(x) - Q(x)|, \quad \textnormal{and} \quad \dkl(P \| Q) = \sum_x p(x) \log\frac{p(x)}{q(x)}
    \end{equation}
    respectively. For continuous distributions $P,Q$ with densities $p,q$ over a common dominating measure, these quantities are instead

    \begin{equation}
        \dtv(P,Q) = \frac12 \int |p(x)-q(x)| \mrm{d}x, \quad \textnormal{and} \quad \dkl(P \| Q) = \int p(x) \log\frac{p(x)}{q(x)} dx.
    \end{equation}
\end{definition}
These two quantities are related by Pinsker's inequality, which states:
\begin{equation}
    \dtv(P,Q) \leq \sqrt{\frac12 \dkl(P \| Q)}.
\end{equation}
Next, we define the Fisher information.
\begin{definition}[Fisher information]
\label{def:fisher-info}
    Let $\{p_s\}$ be a differentiable family of densities with respect to a common dominating measure, such that $s \mapsto p_s(x)$ is absolutely continuous for almost every $x$. We define the score of this distribution with respect to $s$ as
    \begin{equation}
        S_s(x) \coloneqq \frac{\partial \log(p_s(x))} {\partial s},
    \end{equation}
    and the Fisher information is given by
    \begin{equation}
        I_s \coloneqq \mathbb{E}_{x \sim p_s}[S_s(x)^2].
    \end{equation}
\end{definition}

The following standard fact relates the TV distance between two distributions from a parameterized family to the Fisher information and will be crucial to our adaptive lower bounds.

\begin{fact}
\label{fact:tv-and-fisher-info}
    Let $p_s, S_s, I_s$ be as in \Cref{def:fisher-info}, with $I_s < \infty$. Then, for any two parameters $s_0$ and $s_1$,
    \begin{equation}
        \dtv(p_{s_0},p_{s_1}) \leq \frac12 \left|\int_{s_0}^{s_1} \sqrt{I_s} ds\right|.
    \end{equation}
\end{fact}

We will make use of the uniform POVM, a continuous measurement whose measurement operators are $d\ket{\psi}\bra{\psi} \mrm{d}\psi$, where $\mrm{d}\psi$ denotes the Haar measure on $d$-dimensional pure states. For our analysis, we will make use of the following standard fact about the first two moments of such measurements.

\begin{lemma}[Moments of uniform POVM]
  \label{lem:uniform-povm-moments}
  Let $\ket{\psi}$ be the outcome of the uniform POVM applied to a $d$-dimensional state $\rho$. Let $X = (d+1)\ket{\psi}\bra{\psi} - I$. Then,
  \begin{equation}
    \mathbb{E}[X] = \rho \quad \textnormal{and} \quad \mathbb{E}[X^{\otimes 2}] = \frac{1}{d+2}(I^{\otimes 2} + \rho\otimes I + I \otimes \rho )((d+1)\cdot \swap - I).
  \end{equation}
\end{lemma}

We omit the proof of this fact, but note that it follows, with a little work, from the following lemma about moments of Haar-random states.

\begin{lemma}[Haar-measure moments]
    \label{lem:haar-moments}
    Let $r \in \mathbb{N}$, and let $\ket{v} \sim \mathbb{C}^d$ be a $d$-dimensional Haar-random state. Then,
    \begin{equation}
        \mathbb{E}_v(|v\rangle\langle v|)^{\otimes r} =\frac{\sum_{\pi\in S_r}W_\pi}{d(d+1)\cdots(d+r-1)},   
    \end{equation}
    where $W_\pi$ permutes the $r$ tensor factors according to the permutation $\pi$.
\end{lemma}

%% file: sections/3_nonadaptive_estimator.tex
\section{Non-adaptive estimator}
\label{sec:upper-non-adaptive}

We now restate and prove our non-adaptive upper bound from \Cref{thm:intro-non-adaptive}.

\begin{theorem}[Non-adaptive upper bound]
\label{thm:na-upper-purity}
For $d\geq2$, $0<\varepsilon\leq1/4$, there is a non-adaptive single-copy $(d,\varepsilon)$-purity estimator that uses at most
\begin{equation}
    \bigo\Bigl(\max\Bigl\{\frac{d^{3/2}}{\varepsilon},
          \frac{d}{\varepsilon^2}\Bigr\}\Bigr)
\end{equation}
copies of the unknown state.
\end{theorem}

We now present our estimator. We choose $B$ independent Haar-random orthonormal bases and measure $s\geq2$ fresh copies in each basis. For basis $b$, let $N_{b,i}$ denote the number of occurrences of outcome $i\in\{1,\ldots,d\}$. Define the purity estimate by
\begin{equation}
    \widehat P_2 :=\frac1B\sum_{b=1}^B\left((d+1)\widehat Q_b-1\right),
 \label{na-upper:estimator}
\end{equation}
where
\begin{equation}
    \widehat Q_b :=\frac{\sum_{i=1}^dN_{b,i}(N_{b,i}-1)}{s(s-1)}
\end{equation}
is the empirical collision probability in basis $b$.

We will show that our estimator is unbiased and prove the following bound on its variance:
\begin{proposition}[Uniform variance bound]
\label{na-upper:variance-bound}
For all $d\geq2$, $s\geq2$, and $B\geq1$, the estimator in
\Cref{na-upper:estimator} satisfies
\begin{equation}
    \operatorname{Var}(\widehat P_2)
 \leq\frac1B\left(\frac{12d}{s^2}+\frac{12x}{s}+\frac{4x^2}{d}\right).
\end{equation}
\end{proposition}

Now, for a sufficiently large absolute constant $A$, we choose
\begin{equation}
B=\left\lceil A\max\left\{1,\frac1{d\varepsilon^2}\right\}\right\rceil, \qquad 
s=\left\lceil A\min\left\{\frac{d^{3/2}}{\varepsilon},d^2\right\}\right\rceil. \label{na-upper:parameters} 
\end{equation}
Using these parameter choices and our variance bound above, we can prove \Cref{thm:na-upper-purity}.

\begin{proof}[Proof of~\Cref{thm:na-upper-purity}]
We first analyze a single instance of the estimator. Since $P_2=1/d+x$, we have $P_2^{-2}\leq d^2$ and $x^2/P_2^2\leq1$. Also, $x/P_2^2=d^2x/(1+dx)^2\leq d/4$. \Cref{na-upper:variance-bound} therefore implies
\begin{equation}
 \frac{\operatorname{Var}(\widehat P_2)}{P_2^2}
 \leq\frac1B\left(\frac{12d^3}{s^2}+\frac{3d}{s}+\frac4d\right).
 \label{na-upper:relative-variance}
\end{equation}
Choose $B$ and $s$ as in \Cref{na-upper:parameters}. If $\varepsilon\geq d^{-1/2}$, then $B\geq A$ and $s\geq A d^{3/2}/\varepsilon$. Using $\varepsilon/\sqrt d\leq\varepsilon^2$ and $1/d\leq\varepsilon^2$ in \Cref{na-upper:relative-variance} gives
\begin{equation}
    \frac{\operatorname{Var}(\widehat P_2)}{P_2^2} \leq\left(\frac{12}{A^3}+\frac3{A^2}+\frac4A\right) \varepsilon^2.
\end{equation}
 
If $\varepsilon<d^{-1/2}$, then
$B\geq A/(d\varepsilon^2)$ and $s\geq A d^2$, which yield the
same bound. Choose $A$ so that the coefficient is at most $1/4$.
Unbiasedness and Chebyshev's inequality imply
\begin{equation}
    \Pr\!\left[|\widehat P_2-P_2|>\varepsilon P_2\right]
 \leq\frac{\operatorname{Var}(\widehat P_2)}{\varepsilon^2P_2^2}\leq\frac14.
\end{equation}
Each instance uses $Bs=\bigo(d^{3/2}/\varepsilon)$ copies when
$\varepsilon\geq d^{-1/2}$, and $Bs=\bigo(d/\varepsilon^2)$ otherwise.
Equivalently,
\begin{equation}
    Bs=\bigo\!\Bigl(\max\Bigl\{\frac{d^{3/2}}{\varepsilon}, \frac{d}{\varepsilon^2}\Bigr\}\Bigr),
\end{equation}
as claimed.
\end{proof}

It remains now to show that our estimator is unbiased and to bound its variance. There are two sources of randomness: the outcomes in a fixed basis and the choice of the random basis itself. Our analysis will control the variance contributions from each source. For this purpose, it will be convenient to define $\Delta:=\rho-I/d$ and let
\begin{equation}
    x:=\tr(\Delta^2)=P_2-1/d.
\label{na-upper:centered-moments}
\end{equation}
We also write $t_j:=\tr(\Delta^j)$ for $j\in\{3,4\}$.

For a basis represented by a unitary $U$, let $q_i:=\langle i|U^\dagger\rho U|i\rangle$ denote the outcome probabilities and write $Q_j:=\sum_{i=1}^d q_i^j$, the $j$th power sum probability in this basis. When $U$ is Haar-random, define $\mu_j:=\mathbb{E}_U Q_j$ for $j\in\{2,3\}$ and $\mu_{22}:=\mathbb{E}_U Q_2^2$. We start by analyzing the conditional statistics of the collision estimator in a fixed basis.

\begin{lemma}[Collision statistics in a fixed basis]
\label{na-upper:collision-moments}
For any fixed basis $U$ and $s\geq2$ measurement outcomes, the single-block collision estimator $\widehat Q$ satisfies $\mathbb{E}[\widehat Q\mid U]=Q_2$ and
\begin{equation}
 \operatorname{Var}(\widehat Q\mid U)
 =\frac{2(Q_2-Q_2^2)}{s(s-1)}
  +\frac{4(s-2)(Q_3-Q_2^2)}{s(s-1)}.
 \label{na-upper:conditional-variance}
\end{equation}
\end{lemma}

\begin{proof}
Conditional on $U$, let $X_1,\ldots,X_s$ be the independent measurement outcomes, with distribution $(q_1,\ldots,q_d)$. Then
\begin{equation}
    \widehat Q=\binom{s}{2}^{-1}
 \sum_{1\leq a<c\leq s}\mathbf{1}\{X_a=X_c\}.
\end{equation}
Each indicator has expectation $Q_2$ and variance $Q_2-Q_2^2$. Two distinct indicators have covariance zero for disjoint pairs, and $Q_3-Q_2^2$ for pairs sharing one index. There are $\binom{s}{2}$ indicators and $3\binom{s}{3}$ unordered pairs sharing an index. Summing their variances and twice their covariances, then dividing by $\binom{s}{2}^2$, proves
\Cref{na-upper:conditional-variance}.
\end{proof}

Next, we compute the first few moments with respect to the Haar-random unitary, and show that our estimator is indeed unbiased.

\begin{lemma}[Haar averages and unbiasedness]\label{na-upper:haar-averages}
For a Haar-random basis, $\mu_2=(1+P_2)/(d+1)$ and $\mu_3=(1+3P_2+2P_3)/((d+1)(d+2))$. Consequently, the purity estimator is unbiased: $\mathbb{E}\widehat P_2=P_2$.
\end{lemma}

\begin{proof}
For a Haar-random unit vector $v$, we use the identity
\begin{equation}
 \mathbb{E}_v\prod_{j=1}^r\langle v|A_j|v\rangle
 =\frac{\displaystyle\sum_{\pi\in S_r}
       \prod_{c\in\operatorname{cyc}(\pi)}
       \operatorname{tr}\!\left(\prod_{j\in c}A_j\right)}
       {d(d+1)\cdots(d+r-1)},
 \label{na-upper:haar-identity}
\end{equation}
which follows from \Cref{lem:haar-moments}.

Set $a:=\langle v|\rho|v\rangle$. Applying \Cref{na-upper:haar-identity} gives
\begin{equation}
    \mathbb{E}a^2=\frac{1+P_2}{d(d+1)}
\end{equation}
and
\begin{equation}
    \mathbb{E}a^3=\frac{1+3P_2+2P_3}{d(d+1)(d+2)}.
\end{equation}
By symmetry of the Haar columns, $\mu_2=d \cdot \mathbb{E}a^2$ and $\mu_3=d \cdot \mathbb{E}a^3$, proving the stated formulas. \Cref{na-upper:collision-moments} and the law of total expectation give $\mathbb{E}\widehat Q_b=\mu_2$ for every block $b$. Hence $\mathbb{E}\widehat P_2=(d+1)\mu_2-1=P_2$.
\end{proof}

Next, we compute the variance of the collision probability in a Haar-random basis.

\begin{lemma}[Haar fluctuations]
\label{na-upper:haar-fluctuations}
\label{na-upper:moments}
For a Haar-random basis, the variance of its collision probability is
\begin{equation}
 \mu_{22}-\mu_2^2
 =\frac{2\bigl[d(d+1)t_4+(d^2+3d+3)x^2\bigr]}
 {d(d+1)^2(d+2)(d+3)}.
 \label{na-upper:basis-variance}
\end{equation}
The averaged covariance contribution from overlapping sample pairs is
\begin{equation}
 \mu_3-\mu_{22}
 =\frac{2d(d+3)t_3-2dt_4-(d^2+6d+6)x^2+(d^2+5d+6)x}
 {d(d+1)(d+2)(d+3)}.
 \label{na-upper:overlap-moment}
\end{equation}
\end{lemma}

\begin{proof}
Let $v$ be the first column of $U$ and put $a:=\langle v|\rho|v\rangle$ and $r_v:=\langle v|\rho^2|v\rangle$. By symmetry, $\mathbb{E}a^2=\mu_2/d$ and $\mathbb{E}a^3=\mu_3/d$, with $\mu_2$ and $\mu_3$ given by~\Cref{na-upper:haar-averages}. Applying \Cref{na-upper:haar-identity} gives the higher moments
\begin{align}
 \mathbb{E}a^4&=\frac{1+6P_2+3P_2^2+8P_3+6P_4}
 {d(d+1)(d+2)(d+3)},\\
 \mathbb{E}(a^2r_v)&=\frac{P_2+P_2^2+2P_3+2P_4}
 {d(d+1)(d+2)}.
\end{align}
The final identity uses $(A_1,A_2,A_3)=(\rho,\rho,\rho^2)$.

To calculate $\mu_{22}$, let $w$ be a second Haar column, orthogonal to $v$, and put $b:=\langle w|\rho|w\rangle$. Conditional on $v$, the vector $w$ is Haar-uniform on $v^\perp$. For $R:=I-|v\rangle\langle v|$, the second-moment identity in dimension $d-1$ therefore gives
\begin{align}
 \mathbb{E}[b^2\mid v]
 =\frac{[\tr(R\rho R)]^2 +\tr((R\rho R)^2)}{d(d-1)}
 =\frac{(1-a)^2+P_2-2r_v+a^2}{d(d-1)}.
\end{align}
Consequently,
\begin{equation}
    \mathbb{E}(a^2b^2) =\frac{(1+P_2)\mathbb{E}a^2-2\mathbb{E}a^3 +2\mathbb{E}a^4-2\mathbb{E}(a^2r_v)}{d(d-1)}.
\end{equation}
Using $\mu_{22}=d\mathbb{E}a^4+d(d-1)\mathbb{E}(a^2b^2)$, we obtain
\begin{equation}
 \mu_{22}
 =\frac{d^2+4d+2+(2d^2+8d)P_2
 +(d^2+6d+6)P_2^2-8P_3+2dP_4}
 {d(d+1)(d+2)(d+3)}.
 \label{na-upper:haar-fourth-moment}
\end{equation}
Expanding $\rho=I/d+\Delta$ gives $P_2=1/d+x$ and
$P_3=1/d^2+3x/d+t_3$. Likewise,
$P_4=1/d^3+6x/d^2+4t_3/d+t_4$.
Substituting these identities into the formulas for $\mu_2$ and $\mu_3$
in~\Cref{na-upper:haar-averages} and into
\Cref{na-upper:haar-fourth-moment} gives
\Cref{na-upper:basis-variance, na-upper:overlap-moment}.
\end{proof}

We now combine these ingredients to prove our main variance bound.

\begin{proof}[Proof of \Cref{na-upper:variance-bound}]
The eigenvalues of $\Delta$ have absolute value at most one, giving $|t_3|\leq x$. Also, $0\leq t_4\leq x^2$ because the sum of their fourth powers is at most the square of the sum of their squares. Dropping the nonpositive $t_4$ and $x^2$ terms in the numerator of \Cref{na-upper:overlap-moment} and using $t_3\leq x$ yields
\begin{equation}
    0\leq(d+1)^2(\mu_3-\mu_{22})\leq\frac{(3d+2)(d+1)}{d(d+2)}x\leq3x.
\end{equation}
Nonnegativity also follows pointwise from
$(\sum_i q_i^2)^2\leq(\sum_iq_i^3)(\sum_iq_i)=Q_3$.
Similarly, \Cref{na-upper:basis-variance} gives
\begin{equation}
    (d+1)^2(\mu_{22}-\mu_2^2)
 \leq\frac{2(2d^2+4d+3)x^2}{d(d+2)(d+3)}
 \leq\frac{4x^2}{d}.
\end{equation}
Since $\mu_2\leq2/(d+1)$ and $s(s-1)\geq s^2/2$,
\begin{equation}
    \frac{2(d+1)^2(\mu_2-\mu_{22})}{s(s-1)}
 \leq\frac{8(d+1)}{s^2}\leq\frac{12d}{s^2}.
\end{equation}
The law of total variance, \Cref{na-upper:collision-moments}, and independence of the blocks yield
\begin{align}
 \operatorname{Var}(\widehat P_2)
 =\frac{(d+1)^2}{B}\left(
 \frac{2(\mu_2-\mu_{22})}{s(s-1)}
 +\frac{4(s-2)(\mu_3-\mu_{22})}{s(s-1)}+\mu_{22}-\mu_2^2\right).
\end{align}
Combining the preceding bounds with
$(s-2)/(s(s-1))\leq1/s$ proves the proposition.
\end{proof}

%% file: sections/4_adaptive_estimator.tex
\section{Adaptive estimator}
\label{sec:upper-adaptive}

\begin{theorem}
  \label{thm:adaptive}
  There exists an algorithm that succeeds at $(d,\eps)$-purity estimation with adaptive incoherent measurements on
  \begin{equation}
    \bigo\left(
      \frac{d^2}{\varepsilon} + \frac{\sqrt{d}}{\varepsilon^2}
    \right)
  \end{equation}
  copies of the unknown state.
\end{theorem}

We now present the algorithm achieving this upper bound.

\begin{algorithm}
  \caption{Adaptive two-stage purity estimator}
  \label{alg:adaptive}
  \begin{algorithmic}[1]
    \For{$i = 1$ to $n$}
      \State Apply the uniform POVM to a copy of $\rho$ and record outcome $\ket{\psi_i}$.
      \State Compute $X_i = (d+1)\ket{\psi_i}\bra{\psi_i} - I$.
    \EndFor
    \State Compute $\hat\rho \coloneqq \frac1n \sum_{i = 1}^n X_i$
    \State Measure $m$ copies of $\rho$ in the eigenbasis of $\hat{\rho}$ and estimate $\tr(\rho\hat{\rho})$; let the estimated expectation value be $\overline{Y}$.
    \State \Return $Z \coloneqq \frac{2\overline{Y} - \tr(\hat{\rho}^2) + D/n}{1 + 1/n}$. 
  \end{algorithmic}
\end{algorithm}

We show that the adaptive estimator $Z$ constructed above is unbiased and bound its variance in the following lemma.

\begin{lemma}
  \label{lem:adaptive-estimator-variance}
  The estimator $Z$ defined in $\Cref{alg:adaptive}$ satisfies
  \begin{equation}
    \label{eq:adaptive-unbiasedness}
    \mathbb{E}[Z] = \tr(\rho^2),
  \end{equation}
  and
  \begin{equation}
    \label{eq:adaptive-variance}
    \mathrm{Var}[Z] \leq \frac{24d^2}{n^2} + \frac{4 P_3}{m} + \frac{8d}{mn}.
  \end{equation}
\end{lemma}

Before proving the above lemma, we use it to prove our main upper bound.

\begin{proof}[Proof of \Cref{thm:adaptive}]
  Let $C > 0$ be a suitably large constant, and take $n = \left\lceil C \frac{d^2}{\varepsilon}\right\rceil$ and $m = \left\lceil C \left(\frac{\sqrt{d}}{\varepsilon^2} + \frac{d}{\varepsilon}\right)\right\rceil$. \Cref{alg:adaptive} consumes $n + m$ copies, which has the stated copy complexity. It remains to show correctness, for which it suffices to bound the variance.
  By \Cref{lem:adaptive-estimator-variance}, we have
  \begin{align}
    \mathrm{Var}[Z] &\leq \bigo\left(
      \frac{d^2}{n^2} + \frac{P_3}{m} + \frac{d}{mn}
    \right)
    \\&\leq \bigo\left(\frac{\varepsilon^2}{d^2} + \frac{P_2^{3/2} \varepsilon^2}{\sqrt{d}} + \frac{\varepsilon^2}{d^2}\right)
    \\&\leq \bigo\left(\varepsilon^2 P_2^2\right). \label{eq:adaptive-theorem-variance-bound}
  \end{align}
  In the second inequality, we use $n \geq Cd^2/\eps$, $m \geq C\sqrt{d}/\eps^2$ and $P_3 \leq P_2^{3/2}$, and then $n \geq Cd^2/\eps$ and $m \geq Cd/\eps$. In the final step, we use the fact that all $d$-dimensional states have purity at least $\frac1d$.

  Now, by \Cref{lem:adaptive-estimator-variance}, $Z$ is an unbiased estimator for the purity. Taking $C$ sufficiently large and applying Chebyshev's inequality, \Cref{eq:adaptive-theorem-variance-bound} implies
  \begin{equation}
    \mathbf{Pr}[|Z - \tr(\rho^2)| \geq \eps \cdot \tr(\rho^2)] \leq \frac13,
  \end{equation}
  proving correctness.
\end{proof}

It only remains to prove \Cref{lem:adaptive-estimator-variance}, which we do next.
\begin{proof}[Proof of \Cref{lem:adaptive-estimator-variance}]
  We will first show that our estimator is unbiased. By \Cref{lem:uniform-povm-moments}, we have $\mathbb{E}[\hat{\rho}] = \mathbb{E}[X] = \rho$. By direct calculation, $\tr(X^2) = d^2 + d - 1$.
  Consequently,
  \begin{align}
    \mathbb{E}[\tr(\hat{\rho}^2)] &= \frac{1}{n^2} \sum_{i = 1} \mathbb{E}[\tr(X_i^2)] + \sum_{i \neq j} \frac{1}{n^2} \mathbb{E}[\tr(X_i X_j)]
    \\&= \frac{D}{n} + \frac{n-1}{n} \tr(\mathbb{E}[X]^2)
    \\&= \frac{D}{n} + \frac{n-1}{n} P_2,
  \end{align} 
  where the second line follows by independence of $X_i$ and $X_j$. Moreover, as $\mathbb{E}[\hat{\rho}] = \rho$, we have $\mathbb{E}[\overline{Y}] = P_2$. 
  Substituting these into the expectation of our estimator $Z$, we get
  \begin{align}
    \mathbb{E}[Z] &= \frac{2 \mathbb{E}[\overline{Y}] - \mathbb{E}[\tr(\hat{\rho}^2)] + D/n}{1 + 1/n}
    \\&= \frac{P_2 \cdot \left(2 - \frac{n-1}{n}\right) - \frac{D}{n} + \frac{D}{n}}{1 + 1/n} \\&= P_2,
  \end{align}
  as desired.
  
  We will now bound the variance of our estimator. By the law of total variance, we have
  \begin{equation}
  \label{eq:adaptive-variance-split}
      \mathrm{Var}[Z] = \mbb{E}[\mathrm{Var} [Z | \hat{\rho}]] + \mathrm{Var}[\mbb{E}[Z | \hat{\rho}]].
  \end{equation}
  We start by bounding the former term. By definition, we have
  \begin{equation}
      \mathrm{Var} [Z | \hat{\rho}] = \frac{4}{(1+1/n)^2} \mathrm{Var} [\overline{Y} | \hat{\rho}].
  \end{equation}
  When estimating an observable $O$ on a single copy of state $\rho$, the variance is given by $\tr(\rho O^2) - \tr(\rho O)^2$. Averaging over the $m$ copies, we get
  \begin{equation}
      \mathrm{Var} [\overline{Y} | \hat{\rho}] = \frac{\tr(\rho \hat{\rho}^2) - \tr(\rho \hat{\rho})^2}{m} \leq \frac{\tr(\rho \hat{\rho}^2)}{m}. 
  \end{equation}
  Consequently, 
  \begin{equation}
      \mbb{E}[\mathrm{Var} [Z | \hat{\rho}]] \leq \frac{4}{m} \mathbb{E}[\tr(\rho \hat{\rho}^2)].
  \end{equation}
  By independence of the uniform POVM outcomes, we have
  \begin{align}
      \mathbb{E}[\hat{\rho}^2] = \frac1n \mathbb{E}[X^2] + \frac{n-1}{n} \rho^2.
  \end{align}
  Note that $X^2 = (d^2-1) \ket{\psi}\bra{\psi} + I$, and so by rearranging terms in \Cref{lem:uniform-povm-moments}, we get
  \begin{equation}
      \mathbb{E}[X^2] = (d-1) \rho + d \cdot I.
  \end{equation}
  Thus,
  \begin{equation}
      \mathbb{E}[\tr(\rho \hat{\rho}^2)] = \frac{(d-1) P_2 + d}{n} + \frac{n-1}{n} P_3 \leq \frac{2d}{n} + P_3 \implies \mbb{E}[\mathrm{Var} [Z | \hat{\rho}]] \leq \frac{8d}{nm} + \frac{4P_3}{m}.
      \label{eq:adaptive-variance-pt1}
  \end{equation}

  We will now bound the second term, i.e., $\mathrm{Var}[\mbb{E}[Z | \hat{\rho}]]$. The inner conditional expectation is given by
  \begin{equation}
      \mathbb{E}[Z | \hat{\rho}] = \frac{2\tr(\rho \hat{\rho}) - \tr(\hat{\rho}^2) + D/n}{1 + 1/n} = \frac{\tr(\rho^2) - \|\rho - \hat{\rho}\|_2^2 + D/n}{1+1/n}.
  \end{equation}
  Thus,
  \begin{equation}
      \mathrm{Var}[\mbb{E}[Z | \hat{\rho}]] \leq \mathrm{Var}[\|\rho - \hat{\rho}\|_2^2].
  \end{equation}

  To bound this variance, it will be convenient to define $W = X - \rho$, with $W_1, \dots, W_n$ independent copies of this random variable associated with the uniform POVM outcomes. Note that $W$ is traceless and has mean zero. We will view $W$ as a vector in the Hilbert space of traceless Hermitian operators. On this space, we define the covariance operator 
  \begin{equation}
      \mathcal{C}(A) \coloneqq \mathbb{E} [ W \tr(WA)], \quad \forall A \in \mathbb{C}^{d \times d}, \tr(A) = 0, A = A^\dag.
  \end{equation}
  Returning to the variance, we have
  \begin{align}
      \mathrm{Var}[\|\rho - \hat{\rho}\|_2^2] &= \mathrm{Var}\left[
        \frac{1}{n^2} \sum_i \|W_i\|_2^2 + \frac{2}{n^2} \sum_{i < j} \tr(W_i W_j)
      \right]
      \\&= \frac{1}{n^3} \mathrm{Var}[\|W\|_2^2] + \frac{4 \binom{n}{2}}{n^4} \mathrm{\Var}[\tr(W_1 W_2)]
      \\&= \frac{1}{n^3} \mathrm{Var}[\|W\|_2^2] + \frac{2(n-1)}{n^3} \mathrm{\Var}[\tr(W_1 W_2)],
      \label{eq:adaptive-variance-split-again}
  \end{align}
  where the penultimate equality uses the fact that the $W_i$s being centered implies that each pair of terms has zero correlation.

  Let us first bound the first term, i.e., $\mathrm{Var}[\|W\|_2^2]$. We rewrite $\|W\|_2^2$ as follows.
  \begin{equation}
      \|W\|_2^2 = \tr(W^2) = \tr(X^2) + \tr(\rho^2) - 2\tr(\rho X) = D + P_2 - 2\tr(\rho X),
  \end{equation}
  where recall that $\tr(X^2)$ is deterministic and equals $D$.
  Consequently,
  \begin{equation}
      \mathrm{Var}[\|W\|_2^2] = 4 \mathrm{Var}[\tr(\rho X)] \leq 4\mathbb{E}[\tr(\rho X)^2].
  \end{equation}
  To bound the latter term, we use \Cref{lem:uniform-povm-moments} to obtain
  \begin{align}
      \mathbb{E}[\tr(\rho X)^2] &= \tr(\rho^{\otimes 2} \mathbb{E}[X^{\otimes 2}])
      \\& \leq \frac{d+1}{d+2} \tr((\rho^{\otimes 2} + \rho^2 \otimes \rho + \rho \otimes \rho^2) \cdot \swap)
      \\&\leq 3\tr(\rho^2).
  \end{align}
  Thus,
  \begin{equation}
      \mathrm{Var}[\|W\|_2^2] \leq 12 \tr(\rho^2).
      \label{eq:adaptive-variance-3}
  \end{equation}
  We will now bound the latter term in \Cref{eq:adaptive-variance-split-again}. Let $\{E_\alpha\}$ be an orthonormal basis for traceless Hermitian matrices, and let $ W = \sum_\alpha w_\alpha E_\alpha$ be the decomposition of $W$ in this basis. Then, returning to the covariance operator, we can write
  \begin{equation}
      \mathcal{C}_{\alpha,\beta} = \mathbb{E}[\tr(W E_\alpha) \tr(W E_\beta)] = \mathbb{E}[w_\alpha w_\beta].
  \end{equation}
  Now, 
  \begin{equation}
      \mathbb{E}[\tr(W_1 W_2)^2] = \mathbb{E}\left[\left(\sum_\alpha w_{1,\alpha} w_{2,\alpha}\right)^2\right] = \sum_{\alpha,\beta} \mathbb{E}[w_{1,\alpha} w_{1,\beta}] \mathbb{E}[w_{2,\alpha} w_{2,\beta}] = \sum_{\alpha,\beta} \mathcal{C}_{\alpha, \beta}^2 = \tr (\mathcal{C}^2).
      \label{eq:adaptive-variance-5}
  \end{equation}
  We will now analyze the eigenvalues of the covariance operator $\mathcal{C}$. First, let us bound its trace.
  \begin{equation}
      \tr(\mathcal{C}) = \sum_\alpha \tr(E_\alpha C(E_\alpha)) = \sum_\alpha \mathbb{E} [\tr(W E_\alpha)^2] = \mathbb{E} [\|W\|_2^2].
  \end{equation}
  Recall that $\|W\|_2^2 = D + P_2 - 2\tr(\rho X)$, and $\mathbb{E}[X] = \rho$, implying 
  \begin{equation}
      \tr(\mathcal{C}) = D - P_2 = d^2 + d - 1 - P_2 \leq 2d^2.
  \end{equation}
  Note that $\mathcal{C} \succeq 0$, as $\langle A, \mathcal{C}(A) \rangle = \mathbb{E} [\tr(WA)^2] \geq 0$ for all traceless Hermitian matrices $A$. We will now upper bound the operator norm of $\mathcal{C}$. For any valid $A$,
  \begin{align}
      \langle A, \mathcal{C}(A) \rangle &= \mathbb{E} [\tr(WA)^2]
      \\&= \mathrm{Var} [\tr(XA)] \leq \mathbb{E}[\tr(XA)^2]
      \\&= \tr(A^{\otimes 2} \mathbb{E}[X^\otimes 2])
      \\&= \frac{d+1}{d+2} \tr( A^{\otimes 2} (I^{\otimes 2} + \rho \otimes I + I \otimes \rho)\cdot \swap)
      \\&= \frac{d+1}{d+2} (2\tr(\rho A^2) + \tr(A^2)) \leq 3 \tr(A^2).
  \end{align}
  The second equality uses $W = X - \mathbb{E}[X]$, the third line uses \Cref{lem:uniform-povm-moments} and the fact that $A$ is traceless, and the last line uses $\rho \preceq I$ and that $A^2$ is psd. We have thus shown that $\tr(\mathcal{C}) \leq 2d^2$ and $0 \leq \mathcal{C} \leq 3 I$; together, these imply $\tr(\mathcal{C}^2) \leq 6d^2$. Along with \Cref{eq:adaptive-variance-5} and \Cref{eq:adaptive-variance-split-again}, we get
  \begin{equation}
      \mathrm{Var}[\tr(W_1 W_2)] = \mathbb{E}[\tr(W_1 W_2)^2] = \tr(\mathcal{C}^2) \leq 6d^2.
      \label{eq:adaptive-variance-4}
  \end{equation}
  Now, substituting \Cref{eq:adaptive-variance-3,eq:adaptive-variance-4} into \Cref{eq:adaptive-variance-split-again}, we get
  \begin{equation}
      \label{eq:adaptive-variance-pt2}
      \mathrm{Var}[\|\rho - \hat{\rho}\|_2^2 ] \leq \frac{12 P_2}{n^3} + \frac{12d^2}{n^2} \leq \frac{24d^2}{n^2} \implies \mathrm{Var}[\mathbb{E}[Z | \hat{\rho}]] \leq \frac{24d^2}{n^2}.
  \end{equation}
  Finally, \Cref{eq:adaptive-variance-pt1,eq:adaptive-variance-pt2} and the law of total variance (\Cref{eq:adaptive-variance-split}) imply
  \begin{equation}
      \mathrm{Var}[Z] \leq \frac{24d^2}{n^2} + \frac{4P_3}{m} + \frac{8d}{nm},
  \end{equation}
  concluding the proof of the lemma.
\end{proof}

%% file: sections/5_nonadaptive_lower_bounds.tex
\section{Non-adaptive lower bounds}
\label{sec:lower-non-adaptive}

We now prove a matching lower bound for $(d,\eps)$-purity estimation with incoherent non-adaptive measurements.

\begin{theorem}[Non-adaptive lower bound]\label{thm:na-lower-main}
There exists a universal constant $\varepsilon_0>0$ such that, for all sufficiently large $d$ and every $0<\varepsilon\leq\varepsilon_0$, the copy complexity of $(d,\eps)$-purity estimation with incoherent non-adaptive measurements is at least
\begin{equation}
    \Omega\Bigl(\max\Bigl\{\frac{d^{3/2}}{\varepsilon},\frac{d}{\varepsilon^2}\Bigr\}\Bigr).
     \label{eq:na-lower-main}
\end{equation}
\end{theorem}

The two terms arise from different testing problems. The former term follows immediately from a prior mixedness testing lower bound. To obtain the $d/\varepsilon^2$ term, we compare two states with a weak rank-one spike in an unknown Haar-random direction. Non-adaptivity is essential to this second argument: every measurement must be selected before acquiring any information about the spike direction. Throughout the proof, $\varepsilon_0$ denotes a universal constant chosen small enough for both reductions below. 

We start by proving the bound implied by mixedness testing.

\begin{proposition}[Mixedness contribution]
\label{prop:na-lower-mixedness}
For all sufficiently large $d$ and $0<\varepsilon\leq\varepsilon_0$,
the copy complexity of $(d,\eps)$-purity estimation with single-copy measurements is at least
$\Omega(d^{3/2}/\varepsilon)$. This bound also holds when the measurements are adaptive.
\end{proposition}

\begin{proof}
Chen, Huang, Li, and Liu~\cite{chen2022tight} show that distinguishing $\rho=I/d$ from $\|\rho-I/d\|_1>\eta$ with success probability at least $2/3$ requires
$\Omega(d^{3/2}/\eta^2)$ copies, even with adaptive single-copy POVMs, provided $0<\eta\leq1/12$ and $d$ is sufficiently large.

Set $\eta^2=4\varepsilon$. Choosing $\varepsilon_0\leq1/576$ ensures that $\eta\leq1/12$. Under the null hypothesis, $P_2(\rho)=1/d$. Under the alternative, the Schatten norm inequality
$\|A\|_1\leq\sqrt d\,\|A\|_2$ gives
\begin{equation}
 P_2(\rho)
 =\frac1d+\|\rho-I/d\|_2^2
 \geq\frac1d+\frac{\|\rho-I/d\|_1^2}{d}
 >\frac{1+4\varepsilon}{d}.
 \label{eq:na-lower-mixedness-gap}
\end{equation}
Consequently, a successful multiplicative estimate is at most $(1+\varepsilon)/d$ under the null and strictly larger than $(1-\varepsilon)(1+4\varepsilon)/d$ under the alternative. These ranges are disjoint because
\begin{equation}
    (1-\varepsilon)(1+4\varepsilon)-(1+\varepsilon)
 =2\varepsilon(1-2\varepsilon)>0.
\end{equation}
Thresholding $\widehat P_2$ anywhere between the two ranges solves the mixedness testing problem with the same success probability. The claimed bound follows from $\eta^2=4\varepsilon$.
\end{proof}

We now present our main lower bound for this section, i.e., the $\Omega(d/\eps^2)$ lower bound for non-adaptive algorithms. For a unit vector $v\in\mathbb{C}^d$ and a parameter $0\leq\theta<1$, consider
\begin{equation}
 \rho_{\theta,v}=(1-\theta)\frac Id+\theta|v\rangle\langle v|.
 \label{eq:na-lower-spike-state}
\end{equation}
Its purity is independent of the direction $v$:
\begin{equation}
 P_2(\rho_{\theta,v})
 =\frac1d+\left(1-\frac1d\right)\theta^2.
 \label{eq:na-lower-spike-purity}
\end{equation}
The following lemma bounds the information acquired about $\theta$ by
one measurement chosen independently of $v$.

\begin{lemma}[Information in one measurement of a Haar-random spike]
\label{lem:na-lower-haar-information}
Let $v$ be Haar distributed on the unit sphere of $\mathbb{C}^d$, and let $\{M_y\}_y$ be a fixed POVM. Write $p_{\theta,v}$ for its outcome distribution on $\rho_{\theta,v}$. For all $s,t\in[0,1)$,
\begin{equation}
 \mathbb{E}_v\operatorname{KL}(p_{t,v}\|p_{s,v})
 \leq\frac{(t-s)^2}{1-s}\frac{d-1}{d+1}
 \leq\frac{(t-s)^2}{1-s}.
 \label{eq:na-lower-haar-information}
\end{equation}
\end{lemma}

\begin{proof}
Omit zero POVM effects and set $a_y={\tr(M_y)}/d$, and $b_y(v)=\langle v|M_y|v\rangle$. Then $p_{\theta,v}(y)=a_y+\theta(b_y(v)-a_y)$ and
$p_{s,v}(y)\geq(1-s)a_y>0$. Bounding KL-divergence by $\chi^2$, we get
\begin{align}
\dkl(p_{t,v}\|p_{s,v})
 \leq\sum_y\frac{(p_{t,v}(y)-p_{s,v}(y))^2}{p_{s,v}(y)}\leq\frac{(t-s)^2}{1-s}\sum_y\frac{(b_y(v)-a_y)^2}{a_y}.
 \label{eq:na-lower-one-copy-kl}
\end{align}
The first two Haar moments give $\mathbb{E}_v b_y(v)=\tr(M_y)/d$, and
\begin{equation}
    \mathbb{E}_v b_y(v)^2
 =\frac{\operatorname{tr}(M_y)^2+\operatorname{tr}(M_y^2)}{d(d+1)}.
\end{equation}
It follows that
\begin{align}
 \sum_y\frac{\mathbb{E}_v(b_y(v)-a_y)^2}{a_y}
 &=\frac1{d(d+1)}\sum_y
   \left(\frac{d\operatorname{tr}(M_y^2)}{\operatorname{tr}(M_y)}
               -\operatorname{tr}(M_y)\right)\\
 &\leq\frac1{d(d+1)}\sum_y(d-1)\operatorname{tr}(M_y)
 =\frac{d-1}{d+1}.
\end{align}
The inequality uses positivity of $M_y$, which implies $\tr(M_y^2)\leq\operatorname{tr}(M_y)^2$, and the final equality uses $\sum_y M_y=I$. Taking expectations in~\Cref{eq:na-lower-one-copy-kl} proves the claim.

For a POVM with a general outcome space, take a scalar measure dominating the POVM and replace its effects by their matrix-valued densities. The same argument applies with sums replaced by integrals.
\end{proof}

\begin{proposition}[Non-adaptive precision contribution]\label{prop:na-lower-precision}
For $d\geq16$ and $0<\varepsilon\leq\varepsilon_0$, the copy complexity of $(d,\eps)$-purity estimation with non-adaptive measurements is at least $\Omega(d/\varepsilon^2)$.
\end{proposition}

\begin{proof}
We compare two hypotheses $H_0$ and $H_1$, under which the unknown state is $\rho_{\theta_0,v}$ and $\rho_{\theta_1,v}$, respectively. Here $\theta_0=1/\sqrt d$ and $\theta_1=(1+8\varepsilon)/\sqrt d$. Under either hypothesis, $v$ is drawn once from the Haar measure and held fixed across all copies. Our parameter range ensures $0<\theta_0<\theta_1<1/2$. Let $P_2^{(b)}=P_2(\rho_{\theta_b,v})$ for $b\in\{0,1\}$; these are deterministic quantities, independent of $v$. By
\Cref{eq:na-lower-spike-purity},
\begin{align}
    P_2^{(0)}&=\frac{2-1/d}{d},\\
    P_2^{(1)}-P_2^{(0)}&=\frac{1-1/d}{d}(16\varepsilon+64\varepsilon^2)\geq4\varepsilon P_2^{(0)}.
\end{align}
Thus $(1-\varepsilon)P_2^{(1)}>(1+\varepsilon)P_2^{(0)}$, so the two
multiplicative-error intervals are disjoint. Thresholding a successful
purity estimator therefore distinguishes $H_0$ from $H_1$ with success
probability at least $2/3$ under each hypothesis, including after
averaging over $v$.

Let $R$ contain all classical randomness used to choose the measurements.
Conditioned on $R=r$, the protocol fixes a POVM on each of the $N$ copies;
write $p^{r,i}_{\theta,v}$ for the outcome distribution of its $i$th POVM.
Conditioned on both $r$ and $v$, the outcomes are independent. The
conditional transcript distribution under $H_b$ is consequently
\begin{equation}
    \mathsf{T}_b^r=\mathbb{E}_v\left[\bigotimes_{i=1}^N p^{r,i}_{\theta_b,v}\right].
\end{equation}
Let $\mathsf{T}_b$ denote the joint distribution of $R$ and all measurement
outcomes under $H_b$. The seed has the same distribution under both
hypotheses. Applying the log-sum inequality to the mixtures over $v$,
followed by additivity of KL-divergence for product distributions,
gives
\begin{align}
 \dkl(\mathsf{T}_1\|\mathsf{T}_0)
 &=\mathbb{E}_R\dkl(\mathsf{T}_1^R\|\mathsf{T}_0^R)\\
 &\leq\mathbb{E}_{R,v}\sum_{i=1}^N
       \dkl\left(p^{R,i}_{\theta_1,v}\|p^{R,i}_{\theta_0,v}\right)\\
 &\leq\frac{N(\theta_1-\theta_0)^2}{1-\theta_0}
 \leq\frac{128N\varepsilon^2}{d}.
 \label{eq:na-lower-transcript-kl}
\end{align}
The penultimate inequality is~\Cref{lem:na-lower-haar-information},
applied separately to each fixed POVM, and the last inequality uses
$\theta_0\leq1/2$ and $\theta_1-\theta_0=8\varepsilon/\sqrt d$.

A test succeeding with probability at least $2/3$ under each hypothesis
requires $\dtv(\mathsf{T}_1,\mathsf{T}_0)\geq1/3$. Pinsker's inequality and
\Cref{eq:na-lower-transcript-kl} therefore imply
\begin{equation}
    \frac19\leq\operatorname{TV}(\mathsf{T}_1,\mathsf{T}_0)^2\leq\frac12\operatorname{KL}(\mathsf{T}_1\|\mathsf{T}_0)\leq\frac{64N\varepsilon^2}{d}.
\end{equation}
In particular, $N\geq d/(576\varepsilon^2)$.
\end{proof}

Together, \Cref{prop:na-lower-mixedness,prop:na-lower-precision} immediately imply \Cref{thm:na-lower-main} for a universal $\varepsilon_0 > 0$ small enough for both propositions.

%% file: sections/6_adaptive_lower.tex
\section{Adaptive lower bounds}
\label{sec:lower-adaptive}

We will prove two separate lower bounds for purity estimation with adaptive measurements, the maximum of which yields the lower bound in \Cref{thm:intro-adaptive}. We first state the easier of the two, which follows immediately from prior work.

\begin{proposition}
\label{prop:adaptive-lower-easy}
    There exist absolute constants $\varepsilon_0 > 0$, $d_0 > 1$, such that for $d \geq d_0$ and $0 < \eps \leq \varepsilon_0$, the copy complexity of $(d,\varepsilon)$-purity estimation using incoherent measurements is at least
    \begin{equation}
        \Omega\left(
            \max\left\{
            \frac{d^{3/2}}{\varepsilon}, \frac{\sqrt{d}}{\varepsilon^2}
            \right\}
        \right).
    \end{equation}
\end{proposition}
\begin{proof}
    Recall \Cref{prop:na-lower-mixedness}, which showed that the $\Omega(d^{3/2}/\varepsilon)$ lower bound follows from the mixedness testing lower bound of \cite{chen2022tight}, even for adaptive measurements. For the latter term in the statement, we note that \cite{AISW20} show this $\Omega(\sqrt{d}/\varepsilon^2)$ lower bound even for purity estimation with \emph{fully entangled} measurements, implying the same lower bound for the weaker model of incoherent measurements. 
\end{proof}

We now state our main lower bound.

\begin{theorem}
\label{thm:adaptive-lower-main}
    There exist absolute constants $\varepsilon_0 > 0$, $d_0 > 1$, such that for $d \geq d_0$ and $0 < \eps \leq \varepsilon_0$, the copy complexity of $(d,\varepsilon)$-purity estimation using incoherent measurements is at least
    \begin{equation}
        \Omega\left(
            \min \left\{
            \frac{d}{\varepsilon^2}, \frac{d^2}{\varepsilon}
            \right\}
        \right).
    \end{equation}
\end{theorem}

Together, \Cref{prop:adaptive-lower-easy,thm:adaptive-lower-main} imply that the copy complexity of $(d,\eps)$-purity estimation is at least
\begin{equation}
    \Omega\left(
        \max \left\{
            \frac{d^{3/2}}{\varepsilon}, \frac{\sqrt{d}}{\varepsilon^2}, \min \left\{
                \frac{d}{\varepsilon^2}, \frac{d^2}{\varepsilon}
            \right\}
        \right\}
    \right).
\end{equation}
One can verify that the above lower bound precisely matches that of \Cref{thm:intro-adaptive} in every parameter regime by looking at the regimes established in \Cref{fig:main}.

The rest of this section is devoted to proving \Cref{thm:adaptive-lower-main}. We will prove this lower bound by reducing to a classical testing problem for multivariate Gaussian distributions from noisy linear queries, and then directly analyzing this classical task. Our reduction proceeds in two steps: first, we encode these Gaussian distributions into quantum states such that estimating their purity allows one to solve the testing problem; then, we show that any single-copy POVM on such states can be simulated by a constant expected number of noisy linear queries to the encoded Gaussians. The classical testing lower bound thus immediately implies the same bound (up to constant factors) for purity estimation. 

We present this reduction formally in \Cref{sec:gaussian-reduction}, and prove the hardness of the classical problem in \Cref{sec:gaussian-lower}. We then combine these ingredients to prove \Cref{thm:adaptive-lower-main} in \Cref{sec:adaptive-lower-final}.

\subsection{Gaussian testing reduction}
\label{sec:gaussian-reduction}

We will consider distinguishing between two mixtures of states. To define these mixtures, first pick any orthonormal basis of Hermitian traceless matrices; let these be $V_1, \dots, V_{d^2-1}$. For ease of notation, we let $D \coloneqq d^2-1$. Given any vector $\theta \in \mathbb{R}^D$, we define the operator
\begin{equation}
    \Delta_\theta = \sum_{r = 1}^D \theta_r V_r, \quad \rho_\theta \coloneqq \frac{I}{d} + \Delta_\theta.
\end{equation}

As the $V_r$s are traceless, so is $\Delta_\theta$. The operator $\rho_\theta$ is thus a valid density matrix whenever $\Delta_\theta \succeq -I/d$. Our mixtures will be obtained by sampling random parameter vectors $\bftheta$ from two different Gaussian distributions, and we will show that the resulting operators are valid quantum states with very high probability. 

To define our Gaussian priors, we choose the following parameters:
\begin{equation}
    \label{eq:gaussian-parameters}
    \beta = \min\{1,d\eps\}, \quad v_0 = c_0 \beta d^{-3}, \quad \alpha = \frac{10\eps}{D d v_0}, \quad v_1 = (1+\alpha)v_0,
\end{equation}
where $c_0 < 1$ is a sufficiently small positive constant. Throughout, we take $d \geq d_0$ and $\eps \leq \eps_0$, where $d_0,\eps_0$ are absolute constants chosen such that $\alpha < \frac14$ across all parameter regimes.

Then, we will sample vectors $\bftheta \sim \mc{N}(0, v_i I_D)$ for $i = 0$ or $1$, and consider the task of distinguishing between the two priors. The following lemma implies that the operators $\rho_{\bftheta}$ are likely to be valid quantum states in both settings.

\begin{lemma}
\label{lem:instance-valid-state}
    Let $\bftheta \sim \mc{N}(0, v_i I_D)$, for $i = 0$ or $1$. For an absolute constant $d_0 > 1$ with $d \geq d_0$,
    \begin{equation}
        \mathbf{Pr}\left[\|\Delta_{\bftheta}\|_\infty > \frac{1}{2d}\right] \leq .01.
    \end{equation}
\end{lemma}

\begin{proof}
    For a unit vector $u \in \mathbb{C}^{d}$, define the random variable $X_u \coloneqq \bra{u} \Delta_{\bftheta} \ket{u}$. Let $P_u \coloneqq \ket{u}\bra{u} - \frac{I}{d}$ denote the projection of $\ket{u}\bra{u}$ into the space of traceless Hermitian matrices, and observe that the tracelessness of $\Delta_{\bftheta}$ implies $X_u = \tr(P_u \Delta_{\bftheta})$. 

    As the ${\bftheta}_i$s have mean zero, we have $\mathbb{E}[X_u] = 0$. Thus, 
    \begin{equation}
        \mathrm{Var}[X_u] = \mathbb{E}[X_u^2] = v_i \sum_{r = 1}^D \tr(V_r P_u)^2,
    \end{equation}
    where we use the fact that $\bftheta \sim \mc{N}(0, v_i I_D)$. Now, as the $V_r$s form an orthonormal basis of the space of traceless Hermitian matrices, and $P_u$ belongs to this space, we have
    \begin{equation}
        \mathrm{Var}[X_u] = v_i \|P_u\|_2^2 = v_i \tr((\ket{u}\bra{u} - I/d)^2) = v_i(1-1/d) \leq v_i.
    \end{equation}
    Note that $\alpha < \frac14, \beta \leq 1$ and thus, $v_0,v_1 \leq 2c_0/d^3$, which serves as an upper bound on $\mathrm{Var}[X_u]$.

    We will use a standard covering-net argument to bound $\|\Delta_{\bftheta}\|_\infty$. Let $\mathcal{C}$ be a $1/4$-net of the unit Euclidean ball in $\mathbb{C}^d$, which has cardinality at most $9^{2d}$ (see e.g.~\cite[Corollary 4.2.11]{vershynin2019high} and map $\mathbb{C}^d$ to $\mathbb{R}^{2d}$). By standard arguments,
    \begin{equation}
        \|\Delta_{\bftheta}\|_{\infty} \leq 2 \cdot \max_{u \in \mc{C}} |\bra{u} \Delta_{\bftheta} \ket{u}|.
        \label{eq:operator-norm-to-net}
    \end{equation}
    We will thus use the Gaussian tail bound to upper bound the probability that any such inner product exceeds $1/4d$. In particular, for a fixed $u \in \mc{C}$, applying the Gaussian tail bound to the random variable $X_u$, we get
    \begin{equation}
        \mathbf{Pr}[|X_u| \geq 1/4d] \leq 2 \exp\left(\frac{-1}{32d^2 v_i}\right) \leq 2\exp\left(
            \frac{-d}{64c_0}
        \right).
    \end{equation}
    Consequently, taking a union bound over $\mc{C}$, 
    \begin{equation}
        \mathbf{Pr}[\max_{u \in \mc{C}} |X_u| \geq 1/4d] \leq |\mc{C}| \cdot 2\exp\left(
            \frac{-d}{64c_0}
        \right) \leq \exp\left(
            2d\log 9 - \frac{d}{64c_0}
        \right) \leq .01,
    \end{equation}
    where the inequality holds for sufficiently small $c_0$ and $d \geq d_0$ for an appropriate constant $d_0$. By \Cref{eq:operator-norm-to-net}, the above also serves as an upper bound on $\mathbf{Pr}[\|\Delta_{\bftheta}\|_\infty \geq 1/2d]$, concluding the proof.
\end{proof}

We also show that these operators are likely to have separated purities, meaning that they are valid hard instances for purity estimation.

\begin{lemma}
\label{lem:instance-purity-separation}
    For $i \in \{0,1\}$, there exists a set of $D$-dimensional vectors $\mathcal{V}_i \subseteq \mathbb{R}^D$ such that, for $\bftheta \sim \mc{N}(0, v_i I_D)$,
    \begin{equation}
        \mathbf{Pr}[\bftheta \notin \mathcal{V}_i] \leq .01. 
    \end{equation}
    Moreover, an $\eps$-accurate multiplicative error estimate of $\tr(\rho_\theta^2)$ can distinguish between $\theta \in \mathcal{V}_0$ and $\theta \in \mathcal{V}_1$.
\end{lemma}

\begin{proof}
    For any state $\rho_\theta$, its purity is given by
    \begin{equation}
        \tr(\rho_\theta^2) = \frac{1}{d} + \tr(\Delta_\theta^2) = \frac{1}{d} + \sum_{r = 1}^D \theta_r^2,
    \end{equation}
    as the $V_r$s are orthonormal. Note that $\bftheta \sim \mc{N}(0, v_i I_D)$ is equivalent to independently sampling each ${\bftheta}_r$ from $\mc{N}(0,v_i)$. Thus,
    \begin{equation}
        \mathbb{E}[\tr(\rho_{\bftheta}^2)] = \frac{1}{d} + v_i D.
    \end{equation}
    For ease of notation, let $p_i \coloneqq \frac{1}{d} + v_i D$. Then, note that
    \begin{equation}
        p_1 - p_0 = D(v_1 - v_0) = D\alpha v_0 = \frac{10\eps}{d},
    \end{equation}
    by \Cref{eq:gaussian-parameters}.
    
    Now, using the fact that the ${\bftheta}_r$s are independent, we bound the variance of the purity by
    \begin{equation}
        \mathrm{Var}[\tr(\rho_{\bftheta}^2)] = D \cdot \mathrm{Var}[{\bftheta}_1^2] = 2Dv_i^2,
    \end{equation}
    where we use the fact that a standard Gaussian's fourth moment is $3$ and its second moment is $1$. From \Cref{eq:gaussian-parameters}, we have $v_i \leq 2c_0 \beta d^{-3}$, $\beta \leq d\eps$, and $D \leq d^2$, yielding
    \begin{equation}
        \mathrm{Var}[\tr(\rho_{\bftheta}^2)] \leq \frac{8c_0^2\varepsilon^2}{d^2}.
    \end{equation}
    We will define $\mc{V}_i \coloneqq \{\theta \in \mathbb{R}^D : |\tr(\rho_\theta^2) - p_i| \leq \frac{\eps}{d}\}$. By Chebyshev's inequality,
    \begin{equation}
        \mathbf{Pr}_{\bftheta \sim \mc{N}(0,v_i I_D)}[\bftheta \notin \mc{V}_i] \leq 8c_0^2 \leq .01,
    \end{equation}
    for sufficiently small $c_0$. It remains to show that the two sets are separated in purity.

    For sufficiently small $\eps,c_0$, we have $\tr(\rho_\theta^2) \leq \frac{2}{d}$ whenever $\theta \in \mc{V}_i$. Suppose we have an $\eps$-accurate multiplicative-error estimate $\hat{p}$ of $\tr(\rho_\theta^2)$ for some $\theta \in \mc{V}_i$. Then,
    \begin{equation}
        |\hat{p} - p_i| \leq |\hat{p} - \tr(\rho_\theta^2)| + |\tr(\rho_\theta^2) - p_i| \leq \frac{2}{d} \cdot \eps + \frac{\eps}{d} = \frac{3\eps}{d},
    \end{equation}
    where we bound the first term using the accuracy of the estimator and the second term using membership in $\mc{V}_i$. As $p_0$ and $p_1$ are separated by $10\eps/d$, this estimate can be used to successfully distinguish between $\mc{V}_0$ and $\mc{V}_1$. 
\end{proof}

Consequently, any algorithm for purity estimation can distinguish between the two encoded Gaussian priors. However, instead of directly analyzing this mixture-versus-mixture testing problem at the level of quantum states, we will analyze this in a classical query model; we start by defining such queries.

\begin{definition}
    Let $\theta \in \mathbb{R}^D$ be an unknown vector. An ``$L$-bounded noisy linear query'' to $\theta$ takes as input a vector $a \in \mathbb{R}^D$ with $\|a\|_2 \leq L$ and returns $\bfY = a^\top \theta + \bfZ$, where $\bfZ \sim \mc{N}(0,1)$.
\end{definition}

We note that for an algorithm making multiple queries, the noise $Z$ will be sampled independently for each such query.

The key technical element of our reduction is the following lemma showing that these noisy linear queries can be used to \emph{exactly} simulate any adaptive POVM schedule on the state encoding the unknown vector.

\begin{lemma}
\label{lem:gaussian-simulation}
    Let $\rho_\theta$ be a $d$-dimensional quantum state encoding a vector $\theta \in \mathbb{R}^D$ such that $\|\Delta_\theta\|_\infty \leq 1/2d$. There exists an absolute constant $C > 0$ such that any adaptive schedule of $n$ single-copy measurements on the state $\rho_\theta$ can be simulated using $Cn$ adaptively chosen $d$-bounded noisy linear queries to $\theta$ with probability at least $.99$. Moreover, on failure or on input $\theta$ with $\|\Delta_\theta\|_\infty \geq \frac{1}{2d}$, the simulation terminates at $Cn$ queries and returns an arbitrary output.
\end{lemma}

We defer the proof of the above lemma to \Cref{sec:gaussian-simulation}.

\subsection{Gaussian testing lower bound}
\label{sec:gaussian-lower}

With the reduction in place, we now prove the hardness of the classical testing problem.

\begin{theorem}
\label{thm:classical-gaussian-testing-lower}
    Let $i \in \{0,1\}$, and let $\bftheta \sim \mc{N}(0, v_i I_D)$. For an algorithm making $n$ noisy linear queries to $\bftheta$ with norm bounded by $L$, let $Q_{v_i}$ denote the query-response transcripts. Then, even with adaptive queries,
    \begin{equation}
        \dtv(Q_{v_0},Q_{v_1}) \leq L\sqrt{n}|\sqrt{v_1}-\sqrt{v_0}|.
    \end{equation}
\end{theorem}

We will introduce parameters $s_0,s_1$ with $s_i = \sqrt{v_i}$, reparameterizing the distribution as $\bftheta \sim \mc{N}(0, s_i^2 I_D)$. We will then show the following bound on the Fisher information of the transcript distributions with respect to $s$:

\begin{lemma}
    \label{lem:fisher-information-upper}
    Let $\bftheta \sim \mc{N}(0, s^2 I_D)$ and let $Q_s$ denote the corresponding distribution over transcripts of any adaptive protocol. Then, 
    \begin{equation}
        I_s(Q_s) \leq 4nL^2.
    \end{equation}
\end{lemma}

By \Cref{fact:tv-and-fisher-info}, this would immediately imply \Cref{thm:classical-gaussian-testing-lower}. We will thus spend the rest of this section proving this Fisher information bound. We first show how a single query affects the posterior distribution of the unknown vector.

\begin{lemma}
    \label{lem:posterior-updates}
    Suppose that, conditioned on the current transcript $\mc{F}$, the posterior distribution of $\theta$ is given by $\theta \ | \ \mc{F} \sim \mc{N}(m,C)$, for some mean vector $m$ and covariance matrix $C$. After querying $\theta$ with a vector $a$ and observing
    \begin{equation}
        Y = a^\top \theta + Z, \quad Z \sim \mc{N}(0,1),
    \end{equation}
    define
    \begin{equation}
        b \coloneqq \frac{C a}{\sqrt{1 + a^\top C a}}, \quad \xi \coloneqq \frac{Y - a^\top m}{\sqrt{1 + a^\top C a}}.
    \end{equation}
    Then, conditioned on $\mc{F}$, $\xi \sim \mc{N}(0,1)$ and the updated posterior of $\theta$ is Gaussian with parameters
    \begin{equation}
        m^\prime = m + b\xi, \quad C^\prime = C - bb^\top.
    \end{equation}
    In particular, $C^\prime \preceq C$.
\end{lemma}

\begin{proof}
    Conditioned on $\mc{F}$, the query vector $a$ is deterministic. Then, using Bayes' rule, the posterior density is proportional to
    \begin{equation}
        \exp\left[-\frac{1}{2}(\theta - m)^\top C^{-1} (\theta - m)\right] \exp\left[-\frac12(Y-a^\top \theta)^2\right] \propto \exp\left[
        -\frac12 \left\{
            \theta^\top (C^{-1} + aa^\top) \theta -2(C^{-1}m + Ya)^\top\theta
        \right\}
        \right].
    \end{equation}
    Thus, the posterior is Gaussian with mean $m^\prime$ and covariance $C^\prime$ where
    \begin{equation}
        C^\prime = (C^{-1} + aa^\top)^{-1} \quad \textnormal{and} \quad m^\prime = C^{\prime}(C^{-1}m + Ya).
    \end{equation}
    By the Sherman--Morrison formula, we have
    \begin{equation}
        C^\prime = C - \frac{Caa^\top C}{1 + a^\top C a} = C - bb^\top.
    \end{equation}
    Substituting this into the mean update above, we get
    \begin{equation}
        m^\prime = m + \frac{(Y - a^\top m) Ca}{1 + a^\top Ca} = m + b\xi.
    \end{equation}
    Finally, conditioned on $\mc{F}$,
    \begin{equation}
        a^\top \theta \sim \mc{N}(a^\top m, a^\top C a) \implies Y \sim \mc{N}(a^\top m, 1 + a^\top C a),
    \end{equation}
    proving that $\xi \ | \ \mc{F} \sim \mc{N}(0,1)$.
\end{proof}

Consequently, starting from $\theta \sim \mc{N}(0, s^2 I_D)$, we see that the posterior distribution remains Gaussian after each query. Moreover, each query decreases the covariance matrix. We will use these properties to analyze the score of the transcript law.

\begin{lemma}
\label{lem:score-form}
    For $0 \leq j \leq n$, let $T_j$ denote the adaptive query transcript after the first $j$ queries to a vector $\theta \sim \mc{N}(0, s^2 I_D)$. Using \Cref{lem:posterior-updates}, write
    \begin{equation}
        \theta \ | \ T_j \sim \mc{N}(m_j, C_j).
    \end{equation}
    Let $q_{s,j}$ denote the transcript density. Then, its score is given by
    \begin{equation}
        S_s(T_j) \coloneqq \frac{\partial \log(q_{s,j}(T_j))}{\partial s} = \frac{\|m_j\|_2^2 + \tr(C_j)}{s^3} - \frac{D}{s}.
    \end{equation}
\end{lemma}

\begin{proof}
    The prior density is given by
    \begin{equation}
        p_s(\theta) = (2\pi s^2)^{-D/2} \exp\left(\frac{-\|\theta\|_2^2}{2s^2}\right),
    \end{equation}
    implying its score is
    \begin{equation}
        \frac{\partial \log p_s(\theta)}{\partial s} = \frac{-D}{s} + \frac{\|\theta\|_2^2}{s^3}.
    \end{equation}
    Now, let $K_j(t\ |\ \theta)$ denote the conditional density of the transcript given $\theta$. As the queries only depend on $\theta$, this density is independent of $s$. Then, we can write
    \begin{equation}
        q_{s,j}(t) = \int K_j(t \ | \ \theta) p_s(\theta) \mathrm{d}\theta.
    \end{equation}
    Differentiating under the integral and using the fact that $ \frac{\partial p_s(\theta)}{\partial s} = p_s(\theta) \frac{\partial \log(p_s(\theta))}{\partial s}$, we get
    \begin{equation}
        \frac{\partial q_{s,j}(t)}{\partial s} =  \int K_j(t \ | \ \theta) p_s(\theta) \frac{\partial \log(p_s(\theta))}{\partial s} \mathrm{d}\theta.
    \end{equation}
    We will only consider transcripts with non-zero density and divide the above equation by $q_{s,j}$ to get
    \begin{equation}
         \frac{\partial \log q_{s,j}(t)}{\partial s} = \int \frac{K_j(t \ | \ \theta) p_s(\theta)}{q_{s,j}(t)} \frac{\partial \log(p_s(\theta))}{\partial s} \mathrm{d}\theta = \mathbb{E}_{\theta | t} \left[
            \frac{\partial \log(p_s(\theta))}{\partial s}
         \right],
    \end{equation}
    where the last equality is Bayes' rule. Then, we have
    \begin{equation}
        \frac{\partial \log q_{s,j}(T_j)}{\partial s} = \frac{-D}{s} + \frac{\mathbb{E}[\|\theta\|_2^2 \ | \ T_j]}{s^3} = \frac{-D}{s} + \frac{\|m_j\|_2^2 + \tr(C_j)}{s^3},
    \end{equation}
    proving the lemma.
\end{proof}

Using the above lemmas, we can now prove our main Fisher information upper bound.

\begin{proof}[Proof of \Cref{lem:fisher-information-upper}]
    Fix a round $j$ and condition on the transcript $T_j$ so far. Let the current posterior parameters be $m,C$ and let the updated ones from a queried vector $a$ be $m^\prime, C^\prime$ as in \Cref{lem:posterior-updates}, i.e.,
    \begin{equation}
        m^\prime = m + b\xi \quad \textnormal{and} \quad C^\prime = C - bb^\top,
    \end{equation}
    where we have $\xi \ | \ T_j \sim \mc{N}(0,1)$. By \Cref{lem:score-form}, the score increment from the $(j+1)$th query is given by
    \begin{equation}
        S_s(T_{j+1}) - S_s(T_j) = \frac{\|m^\prime\|_2^2 - \|m\|_2^2 + \tr(C^\prime) - \tr(C)}{s^3} = \frac{2m^\top b \xi + \|b\|_2^2(\xi^2 - 1)}{s^3}.
    \end{equation}
    As $\mathbb{E}[\xi \ | \ T_j] = 0$ and $\mathbb{E}[\xi^2 \ | \ T_j] = 1$, the score increment has conditional mean zero. Moreover, using
    \begin{equation}
        \mathbb{E}[\xi^2 \ | \ T_j] = 1, \quad \mathbb{E}[(\xi^2-1)^2 \ | \ T_j] = 2, \quad \mathbb{E}[\xi(\xi^2-1) \ | \ T_j] = 0,
    \end{equation}
    we have
    \begin{equation}
        \mathbb{E}[(S_s(T_{j+1}) - S_s(T_j))^2 \ | \ T_j] = \frac{4(m^\top b)^2 + 2\|b\|_2^4}{s^6}.
    \end{equation}
    To control the $m^\top b$ term, we introduce auxiliary quantities
    \begin{equation}
        K \coloneqq m^\top C m, \quad K^\prime \coloneqq m^{\prime\top} C^\prime m^\prime.
    \end{equation}
    Using the update rule and the first two conditional moments of $\xi$,
    \begin{align}
        \mathbb{E}[K^\prime \ | \ T_j] &= \mathbb{E}[(m+b\xi)^\top(C-bb^\top)(m+b\xi) \ | \ T_j]
        \\&= K - (m^\top b)^2 + b^\top C b - \|b\|_2^4.
    \end{align}
    Consequently, 
    \begin{align}
        4(m^\top b)^2 + 2 \|b\|_2^4 \leq 4 ((m^\top b)^2 + \|b\|_2^4) = 4 (K - \mathbb{E}[K^\prime \ | \ T_j] + b^\top C b).
    \end{align}
    As $b = \frac{C a}{\sqrt{1 + a^\top C a}}$, we have
    \begin{equation}
        b^\top C b = \frac{a^\top C^3 a}{1 + a^\top C a} \leq a^\top C^3 a \leq s^6 \|a\|_2^2 \leq s^6 L^2,
    \end{equation}
    where we use the fact that $C \preceq s^2 I_D$ (\Cref{lem:posterior-updates}) and the fact that all queries must have norm bounded by $L$. Consequently, 
    \begin{equation}
        \mathbb{E}[(S_s(T_{j+1}) - S_s(T_j))^2 \ | \ T_j] \leq \frac{4}{s^6}(K - \mathbb{E}[K^\prime \ | \ T_j] + s^6L^2).
    \end{equation}
    We will now compute the second moment of the full score. First, note that $S_s(T_0) = 0$, as the empty transcript has no information. Writing $S_s(T_n) = \sum_{j = 0}^{n-1} S_s(T_{j+1}) - S_s({T_j})$, and noting that each increment has conditional mean zero, we get
    \begin{equation}
        \mathbb{E}[S_s(T_n)^2] = \sum_{j = 0}^{n-1} \mathbb{E}[(S_s(T_{j+1}) - S_s(T_{j}))^2].
    \end{equation}
    Writing $K_j = m_j^\top C m_j$, we get
    \begin{equation}
        \mathbb{E}[S_s(T_n)^2] \leq \frac{4}{s^6} (K_0 - \mathbb{E}[K_n] + ns^6L^2) \leq 4nL^2,
    \end{equation}
    where we use $K_n \geq 0$ and that $m_0 = 0 \implies K_0 = 0$. As $I_s(Q_s) = \mathbb{E}[S_s(T_n)^2]$, this concludes the proof.
\end{proof}

As stated before, \Cref{lem:fisher-information-upper} and \Cref{fact:tv-and-fisher-info} immediately yield \Cref{thm:classical-gaussian-testing-lower}, concluding this section.

\subsection{Putting things together}
\label{sec:adaptive-lower-final}

We will now use the ingredients established in the preceding sections to prove \Cref{thm:adaptive-lower-main}.

\begin{proof}[Proof of \Cref{thm:adaptive-lower-main}]
    Suppose there exists a $(d,\eps)$-purity estimator making $n$ adaptively chosen single-copy measurements and succeeding with probability at least $\frac23$. We will use this estimator along with the simulator of \Cref{lem:gaussian-simulation} to solve the Gaussian testing task. In particular, we are given query access to a random $\bftheta \sim \mc{N}(0, v_i I_D)$ and asked to determine $i \in \{0,1\}$. 
    
    To solve this task, we will run the purity estimator and simulate its measurements with \Cref{lem:gaussian-simulation}, making $Cn$ queries to $\theta$. By \Cref{lem:instance-purity-separation}, when $\bftheta \in \mathcal{V}_i$, an accurate estimate can be used to determine $i$. Thus, on all desirable events, our algorithm will succeed at the testing problem. The failure events are 1) $\|\Delta_{\bftheta}\|_\infty \geq \frac{1}{2d}$, 2) $\bftheta \notin \mathcal{V}_i$, 3) failure of the simulator, and 4) failure of the purity estimator. By \Cref{lem:instance-valid-state,lem:instance-purity-separation,lem:gaussian-simulation}, the first three events occur with probability at most $.01$ each, and conditioned on them not occurring, the final event occurs with probability at most $\frac13$. By a union bound, the total failure probability is at most $.37$. 

    However, the success probability of any classical algorithm making such queries is bounded by $\frac12 + \frac12\dtv(Q_0,Q_1)$, where $Q_0,Q_1$ denote the distributions over query transcripts. Consequently, we must have
    \begin{equation}
    \label{eq:tv-distance-1}
        \dtv(Q_0,Q_1) \geq .26.
    \end{equation}
    We will now use \Cref{thm:classical-gaussian-testing-lower} to upper bound this TV distance. By \Cref{lem:gaussian-simulation}, we made a total of $Cn$ queries, each with norm bounded by $d$. From \Cref{thm:classical-gaussian-testing-lower}, we get
    \begin{align}
        \dtv(Q_0,Q_1) &\leq d\sqrt{Cn} |\sqrt{v_1} - \sqrt{v_0}| = d\sqrt{C n v_0} (\sqrt{1+\alpha} - 1)
        \\&\leq \frac{\alpha d\sqrt{C n v_0}}{2}. \label{eq:tv-distance-2}
    \end{align}
    From \Cref{eq:tv-distance-1,eq:tv-distance-2}, we get
    \begin{equation}
        n \geq \Omega\left(\frac{1}{\alpha^2 v_0 d^2}\right).
    \end{equation}
    Recall from \Cref{eq:gaussian-parameters} that $\alpha = 10\eps/Ddv_0$, $v_0 = \Theta(d^{-3} \min\{1,d\eps\}),$ and also that $D = d^2 - 1$. Altogether, we get
    \begin{equation}
        n \geq \Omega\Bigl(\min\Bigl\{\frac{d}{\eps^2}, \frac{d^2}{\eps}\Bigr\}\Bigr),
    \end{equation}
    as desired.
\end{proof}

%% file: sections/A_gaussian_simulation.tex
\section{Gaussian simulation}
\label{sec:gaussian-simulation}

This section is dedicated to proving \Cref{lem:gaussian-simulation}. For the simulation, we will make use of a ``Bernoulli factory'', i.e., an algorithm that takes multiple flips from a coin with success probability $p$ to produce one coin flip with success probability $f(p)$.

\begin{theorem}[Consequence of {\cite[Theorem 2]{nacu2005fast}}]
\label{thm:bernoulli-factory}
    Let $I \subset (0,1)$ be a closed interval and $f: I \to (0,1)$ be a real analytic function on this interval. Define a Bernoulli factory for $f$ as an algorithm that takes samples from a coin with an unknown success probability $p \in I$ and outputs a sample from a coin with success probability exactly $f(p)$. Then, there exists a Bernoulli factory for $f$ with expected number of input samples at most $C_f$ for   
    some fixed $0 < C_f < \infty$ and for any $p \in I$. \footnote{The version stated here is weaker than that of \cite{nacu2005fast}, who actually show that the number of samples has an exponentially decaying tail. However, this weaker version is sufficient for our purposes.}
\end{theorem}

We are now ready to prove \Cref{lem:gaussian-simulation}.

\begin{proof}[Proof of \Cref{lem:gaussian-simulation}]
    For ease of exposition, we only consider POVMs with finitely many outcomes, but the proof naturally extends to continuous POVMs as well. Conditioned on prior measurements and any parameter-independent randomness, the choice of any particular POVM is deterministic. We will thus first show how to simulate a single POVM, and the simulation of the full schedule will follow almost immediately.

    Suppose the POVM in question consists of measurement operators $\{M_y\}_y$. When applied to the maximally mixed state, the outcome distribution would be
    \begin{equation}
        p_0(y) \triangleq \frac{\tr(M_y)}{d}.
    \end{equation}
    However, we apply this to a Gaussian perturbation of the maximally mixed state, i.e., to a state $\rho_\theta = \frac{I}{d} + \Delta_\theta$. For such a state, we write
    \begin{equation}
        p(y) = \tr(M_y \rho_\theta) = p_0(y) + \tr(M_y \Delta_\theta).
    \end{equation}
    We define
    \begin{equation}
        A_y \coloneqq \frac{d \cdot M_y}{\tr(M_y)}, \quad \textnormal{and} \quad \mu(y) \coloneqq \tr(A_y \Delta_\theta),
    \end{equation}
    allowing us to write
    \begin{equation}
        p(y) = p_0(y) (1 + \mu(y)).
    \end{equation}
    In particular, we can view the outcome distribution as a tilted distribution with respect to the ``null'' one, i.e., the one corresponding to the maximally mixed state. As such, our simulation algorithm will draw $y \sim p_0$ and only accept it with a conditional probability. 

    Specifically, we draw $y$ from $p_0$, and then accept this with probability $\frac{1 + \mu(y)}{2}$. Note that $\|\Delta_\theta\|_\infty \leq \frac{1}{2d}$, and so $|\mu(y)| \leq \frac12$ for all $y$, implying that $\frac{1 + \mu(y)}{2} \in [1/4, 3/4]$ and is thus a valid probability. Now, on acceptance, we output $y$; on rejection, we repeat this for a new draw of $y$. The probability of outputting a certain $y$ in a single attempt is thus given by
    \begin{equation}
        \mathbf{Pr}[y \sim p_0 \textnormal{ and } \mathsf{Accept}] = p_0(y) \cdot \frac{1 + \mu(y)}{2} = p(y)/2.
    \end{equation}
    Consequently, the total probability of acceptance is given by
    \begin{equation}
        \sum_y \mathbf{Pr}[y \sim p_0 \textnormal{ and } \mathsf{Accept}] = \sum_y p(y)/2 = \frac12.
    \end{equation}
    Moreover, conditioned on acceptance, we have exactly simulated the desired POVM distribution:
    \begin{equation}
        \mathbf{Pr}[y | \mathsf{Accept}] = \frac{\mathbf{Pr}[y \sim p_0 \textnormal{ and } \mathsf{Accept}]}{\mathbf{Pr}[\mathsf{Accept}]} = p(y).
    \end{equation}
    Thus, these two stages successfully simulate a draw from the POVM, and each attempt has a $\frac12$ success probability.

    However, each attempt of this protocol requires a sample of the Bernoulli random variable with success probability $\frac{1 + \mu(y)}{2}$. As $\mu(y)$ depends on the unknown vector $\theta$, such a sample is not directly available. We will instead obtain such samples using noisy linear queries to $\theta$ along with the Bernoulli factory of \Cref{thm:bernoulli-factory}.

    For the operator $A_y$, let $\Tilde{A}_y \coloneqq A_y - \frac{\tr(A_y)}{d} I$ be its traceless version. By tracelessness of $\Delta_\theta$, we have $\mu(y) = \tr(\Tilde{A}_y \Delta_\theta)$. As $\Tilde{A}_y$ is a traceless Hermitian operator, we can write it in terms of the basis $\{V_r\}_{r \in [D]}$. Let $\{c_r\}$ be the associated coefficients. Then, note that
    \begin{equation}
        \mu(y) =  \tr(\Tilde{A}_y \Delta_\theta) = \sum_{r  = 1}^D c_r \theta_r.
    \end{equation}
    In particular, we can estimate $\mu(y)$ by querying $\theta$ with the vector $c$. Note that
    \begin{equation}
        \|c\|_2 = \|\Tilde{A_y}\|_2 \leq \|A_y\|_2 = d \cdot \frac{\|M_y\|_2}{\|M_y\|_1} \leq d,
    \end{equation}
    where we use the fact that $M_y$ is positive semidefinite. This proves that our queries will have norm bounded by $d$.

    Now, let the output of the query be $Y = \mu(y) + Z$, where $Z \sim \mc{N}(0,1)$. We will treat as our input coin the indicator variable $\mathbf{1}[Y \geq 0]$. This has success probability
    \begin{equation}
        \mathbf{Pr}[Z \geq -\mu(y)] = \Phi(\mu(y)),
    \end{equation}
    where $\Phi$ is the CDF of the standard normal distribution, i.e., $\Phi(x) = \mathbf{Pr}[Z \leq x]$ for $x \in \mathbb{R}$. Feed samples from the above Bernoulli random variable into \Cref{thm:bernoulli-factory} for the function
    \begin{equation}
        f(x) = \frac{1 + \Phi^{-1}(x)}{2}
    \end{equation}
    on the interval $I = [\Phi(-1/2), \Phi(1/2)]$. Note that $\Phi$ is real analytic and strictly increasing on $\mathbb{R}$, so its inverse is real analytic on $(0,1)$, and consequently $\Phi^{-1}(x)$ is real analytic on $I$, allowing us to apply \Cref{thm:bernoulli-factory}. Note that with $x = \Phi(\mu(y))$, we get $f(x) = \frac{1 + \mu(y)}{2}$, exactly the probability desired by the earlier rejection sampler. In expectation, this takes $C_f$ queries to $\theta$. 

    As the success probability of each simulation attempt was $\frac12$, we can simulate a single POVM with $2C_f$ queries in expectation. 

    We will now discuss how to simulate a full measurement schedule. Conditioned on a prior transcript, the choice of the next measurement is deterministic. We apply the simulation routine to each such measurement, and continue until we accept. Additionally, for each POVM, we will add a bookkeeping query to the zero vector to ensure termination under bad events. Thus, the expected number of queries per measurement is $2C_f + 1$.  
    We continue in this way until we succeed or the total number of queries exceeds $100(2C_f + 1)n$, terminating the overall simulation as soon as this occurs. The number of expected queries to succeed is $(2C_f + 1)n$; consequently, the probability of this failure event is at most $.01$ by Markov's inequality. Thus, our algorithm achieves \Cref{lem:gaussian-simulation} for $C = 100 (2C_f+1)$.
\end{proof}